\documentclass{article}

\usepackage{arxiv}

\usepackage[utf8]{inputenc} 
\usepackage[T1]{fontenc}    
\usepackage{hyperref}       
\usepackage{url}            
\usepackage{booktabs}       
\usepackage{amsfonts}       
\usepackage{nicefrac}       
\usepackage{microtype}      
\usepackage{lipsum}
\usepackage{verbatim}
\usepackage{mathrsfs}
\usepackage{amsthm}
\usepackage{blindtext}
\usepackage{tikz}
\usetikzlibrary{decorations.text,calc,arrows.meta}
\usetikzlibrary{backgrounds}
\usepackage{pifont}

\usepackage{amsfonts}
\usepackage{amsmath, amssymb}
\newtheorem{theorem}{Theorem}

\newtheorem{lemma}{Lemma}

\theoremstyle{definition}
\newtheorem{definition}{Definition}

\newtheorem{example}[]{Example}

\title{Algorithmic Techniques for Necessary and Possible Winners}

\author{
    Vishal Chakraborty\\
   \footnotesize UC Santa Cruz, USA\\
   \texttt{vchakrab@ucsc.edu} \\
   \And
   Theo Delemazure \\
   Ecole Normale Superieure, France \\
  \texttt{ theo.delemazure@ens.fr} \\
  \And
  Benny Kimelfeld\\
  Technion - Israel Institute of Technology, Israel\\
  \texttt{bennyk@cs.technion.ac.il}\\
  \And
  Phokion G. Kolaitis\\
  UC Santa Cruz, USA \\
  IBM Research - Almaden, USA\\
  \texttt{kolaitis@ucsc.edu}\\
  \And
  Kunal Relia\\
  New York University, USA\\
  \texttt{krelia@nyu.edu}
  \And
  Julia Stoyanovich\\
  New York University, USA\\
  \texttt{stoyanovich@nyu.edu}\\
}
\usepackage{subcaption}
\usepackage{balance}
\usepackage{multirow}
\usepackage{soul}
\usepackage{xspace}
\usepackage{algorithm}
\usepackage{algorithmic}
\usepackage{relsize}
\usepackage{stmaryrd}

\newcommand{\rsmm}{{\em RSM Mix} }
\newcommand{\julia}[1]{\textcolor{red}{[Julia: #1]}}

\newcommand{\kunal}[1]{\textcolor{red}{[Kunal: #1]}}

\newcommand{\colorOne}{\color{blue}}
\newcommand{\colorTwo}{\color{black}}

\newcommand{\checkNW}{\mathrm{checkNW}}

\newcommand{\PW}{\mathrm{PW}}
\newcommand{\NW}{\mathrm{NW}}

\newcommand{\eat}[1]{}

\newcommand{\up}[1]{$\textsc{Up}_{P}#1$}
\newcommand{\down}[1]{$\textsc{Down}_{P}#1$}
\newcommand{\npcom}{NP-complete\xspace}
\newcommand{\npcomshort}{NP-C\xspace}
\newcommand{\PWPC}{\mbox{PW-PC}}
\newcommand{\PWPV}{\mbox{PW-PP}}
\newcommand{\PWDV}{\mbox{PW-DTB}}
\newcommand{\PWTV}{\mbox{PW-TTB}}
\newcommand{\PWBV}{\mbox{PW-BTB}}
\newcommand{\Ptime}{\mbox{P}}
\newcommand{\NP}{\mbox{\sc NP}}

\newcommand{\plurality}{plurality\xspace}
\newcommand{\veto}{veto\xspace}
\newcommand{\ie}{i.e.,\xspace} 
\newcommand{\eg}{e.g.,\xspace} 
\newcommand{\etal}{et al.\xspace}

\def\btau{\boldsymbol{\tau}} 
\def\bsigma{\boldsymbol{\sigma}}

\newcommand*{\ranking}[1]{\langle #1 \rangle}

\def\RIM{\mathsf{RIM}}
\def\mallows{\mathsf{MAL}}
\def\RSM{\mathsf{RSM}}

\def\Pr{\mathrm{Pr}}
\def\dist{\mathit{dist}}

\def\T{\mathbf{T}}
\def\P{\mathbf{P}}
\begin{document}
\maketitle
\begin{abstract}
We investigate the practical aspects of computing the necessary and possible winners in elections over incomplete voter preferences.  In the case of the necessary winners, we show how to implement and accelerate the polynomial-time algorithm of Xia and Conitzer.  In the case of the possible winners, where the problem is NP-hard, we give a natural reduction to Integer Linear Programming (ILP) for all positional scoring rules and implement it in a leading commercial optimization solver.  Further, we devise optimization techniques to minimize the number of ILP executions and, oftentimes, avoid them altogether.
We conduct a thorough experimental study that includes the construction of a rich benchmark of election data based on real and synthetic data.  Our  findings suggest that, the worst-case intractability of the possible winners notwithstanding, the algorithmic techniques presented here scale well and can be used to compute the possible winners in realistic scenarios.
\end{abstract}
\section{Introduction}
\label{sec:intro}

The theory of social choice focuses on the question of how preferences of individuals can be aggregated in such a way that the society arrives at a collective decision.  It has been of interest throughout the history of humankind, from the analysis of election manipulation by Pliny the Younger in Ancient Rome, to the $18^{th}$ Century studies of voting rules by Jean-Charles de Borda and Marquis de Condorcet, and the more recent ground-breaking work on dictatorial vote aggregation by Kenneth Arrow in the 1950s.  Over the past two decades, \emph{computational social choice} has been developing as an interdisciplinary area between social choice theory, economics, and computer science, where the central topics of study are the computational and algorithmic perspectives of voting challenges such as vote aggregation~\cite{DBLP:reference/choice/0001CELP16}.

A voting rule determines how the collection of voter preferences over a set of candidates is mapped to the set of winning candidates (the winners). Among the most extensively studied is the class of \emph{positional scoring rules}, where every candidate receives a score from every voter that is determined only by the position of the candidate in the voter's ranking. 
A candidate wins if she achieves the highest total score --- the sum of scores it receives from each voter.


A particularly challenging computational aspect arises in situations  in which voter preferences are only \emph{partial} (\ie~they can be modeled as partial orders). This might happen since, for example, voters may be undecided about some candidates or, simply, only partial knowledge of the voter's preference is available (\eg~knowledge is inferred indirectly from opinions on issues).  The problem already manifests itself at the semantic level: what is the meaning of vote aggregation in the presence of incompleteness, if voting rules require complete knowledge?  For this reason, Konczak and Lang~\cite{konczak2005voting} introduced the notions of \emph{necessary winners} and \emph{possible winners} as the candidates who win in \emph{every} completion, and, respectively, \emph{at least one} completion of the given partial preferences. 

This work led to a classification of the computational complexity of the necessary and possible winners for a large variety of voting rules~\cite{DBLP:journals/ipl/BaumeisterR12,DBLP:journals/jcss/BetzlerD10,DBLP:journals/jair/XiaC11}. Specifically, under (efficiently computable) positional scoring rules, the necessary winners can be computed in polynomial time via the algorithm of Xia and Conitzer~\cite{DBLP:journals/jair/XiaC11}.  The possible winners can be computed in polynomial time under the \plurality and \veto rules, but their computation is NP-hard for every other \emph{pure} rule, as established in a sequence of studies~\cite{DBLP:journals/ipl/BaumeisterR12,DBLP:journals/jcss/BetzlerD10,konczak2005voting,DBLP:journals/jair/XiaC11}. Here, \emph{pure} means that the scoring vector for $m$ candidates is obtained from that for $m-1$ candidates by inserting a new score into the vector.

In this paper, we explore the practical aspects of computing the necessary and possible winners. Specifically, we investigate the empirical feasibility of this challenge,  develop algorithmic techniques to accelerate and scale the execution, and  conduct a thorough experimental evaluation of our techniques. 
For the necessary winners, we show how to accelerate the Xia and Conitzer algorithm through  mechanisms of early pruning and early termination.   For the possible winners, we focus on positional scoring rules for which the problem is NP-hard. We first give a natural polynomial-time reduction of the possible winners  to Integer  Linear Programming (ILP) for all positional scoring rules.  Note that ILP has been used in earlier research on the complexity of voting problems as a theoretical technique for proving upper bounds (fixed-parameter tractability) in the parameterized complexity of the possible winners~\cite{DBLP:conf/ijcai/BetzlerHN09,DBLP:conf/pods/KimelfeldKT19,DBLP:conf/ecai/Yang14} or in election manipulation problems involving complete preferences~\cite{DBLP:journals/eor/PolyakovskiyBN16}.  Here, we investigate the use of ILP solvers to compute the possible winners in practice. Our experiments on a leading commercial ILP solver (Gurobi v8.1.1) show that the reduction produces ILP programs that are often too large to load and too slow to solve. For this reason, we develop several techniques to minimize or often eliminate ILP computations and, when the use of ILP is unavoidable, to considerably reduce the size of the ILP programs.

We conduct an extensive experimental study that includes the construction of a rich benchmark of election data based on both real and synthetic data.  Our experimental findings suggest that, the worst-case intractability of the possible winners notwithstanding, the algorithmic techniques presented here scale well and can be used to compute the possible winners in realistic scenarios. 
An important contribution of our work that is of independent interest is a novel generative model for partially ordered sets, called the Repeated Selection Model, or RSM for short.  We believe that RSM may find uses in other experimental studies in computational social choice.


\section{Preliminaries and Earlier Work}
\label{sec:preliminaries}

\paragraph{Voting profiles} 

Let $C = \{c_1, c_2, c_3, \dots, c_m\}$ be a set of \emph{candidates} and let $V = \{v_1,\ldots,v_n\}$ be a set of voters. A \emph{(complete) voting profile} is a tuple $\T=(T_1,\ldots,T_n)$ of total orders of $C$, where each $T_l$ represents the ranking (preference) of voter $v_l$ on the candidates in $C$. Formally, each $T_l$ is a  binary relation $\succ_{T_l}$ on $C$ that is irreflexive (i.e., $c_i \not \succ_{T_l} c_i$, for all $i$), antisymmetric (i.e., $c_i \succ_{T_l} c_j$ implies $c_j \not \succ_{T_l} c_i$, for all $i \not = j)$, transitive (i.e.,    $c_i \succ_{T_l} c_j$ and $c_j \succ_{T_l} c_k$ imply $c_i \succ_{T_l} c_k$, for all $i,j, k$), and total (i.e., $c_i \succ_{T_l} c_j$ or $c_j \succ_{T_l} c_i$ holds for all $i\not = j$).  Similarly, a \emph{partial voting profile} is a tuple $\P=(P_1,\ldots,P_n)$ of partial orders on $C$, where each $P_l$ represents the partial preferences of voter $v_l$ on the candidates in $C$; formally, each $P_l$ is a binary relation on $C$ that is irreflexive, antisymmetric, and transitive (but not necessarily total).   A \emph{completion} of a partial voting profile $\P= (P_1,\ldots,P_n)$ is a complete voting profile $\T=(T_1,\ldots,T_n)$ such that each $T_l$ is a completion of the partial order $P_l$, that is to say, $T_l$ is a total order that extends $P_l$. Note that, in general, a partial voting profile may have exponentially many completions.

\paragraph{Voting rules} We focus on \emph{positional scoring rules}, a widely studied class of voting rules.  A positional scoring rule $r$ on a set of $m$ candidates is specified by a scoring vector ${\bf s}=(s_1,\ldots,s_m)$ of non-negative integers, called the \emph{score values}, such that $s_1\geq s_2\geq \ldots \geq s_m$.  Suppose that $\T=(T_1,\ldots,T_n)$ is a total voting profile. The score $s(T_l,c)$ of a candidate $c$ on $T_l$ is the value $s_k$ where $k$ is the position of candidate $c$ in $T_l$.  The \emph{score} of $c$ under the positional scoring rule $r$ on the total profile $\T$ is the  sum $\sum_{l=1}^ns(T_l,c)$. A candidate $c$ is a \emph{winner} if $c$'s score is greater than or equal to the scores of all other candidates; similarly, $c$ is a \emph{unique winner} if $c$'s score is greater than the scores of all other candidates. The set of all winners is denoted by $\mbox{W}(r, \T)$.

We consider positional scoring rules that are defined for every number $m$ of candidates. Thus, a \emph{positional scoring rule} is an infinite sequence ${\bf s}_1, {\bf s}_2, \ldots, {\bf s}_m, \ldots$ of scoring vectors such that each ${\bf s}_m$ is a scoring vector of length $m$. Alternatively, a positional scoring rule is a function $r$ that takes as argument a pair $(m,s)$ of positive integers with $s\leq m$ and returns as value a non-negative integer $r(m,s)$ such that $r(m,1) \geq r(m,2) \ldots \geq r(m,m)$. We assume that the function $r$ is computable in time polynomial in $m$, hence the winners can be computed in polynomial time. Such a rule is \emph{pure} if the scoring vector ${\bf s}_{m+1}$ of length $(m+1)$ is obtained from the scoring vector ${\bf s}_m$ of length $m$ by inserting a score in some position of  ${\bf s}_ m$, provided that the decreasing order of score values is maintained. We also assume that the scores in every scoring vector are co-prime (i.e., their greatest common divisor is $1$), since multiplying all scores by the same value does not change the winners.

As examples, the \emph{\plurality} rule is given by scoring vectors of the form $(1,0,\dots,0)$, while  the \emph{\veto} rule is given by scoring vectors of the form $(1,1,\ldots,1,0)$. The \plurality rule is the special case $t=1$ of the \emph{$t$-approval} rule, in which the scoring vectors start with $t$ ones and then are followed by zeros. In particular, the \emph{$2$-approval} rule has scoring vectors of the form $(1,1,0,\ldots,0)$. The \emph{Borda} rule, also known as the \emph{Borda count}, is given by scoring vectors of the form $(m-1,m-2,\dots,0)$.

\paragraph{Necessary and possible winners}
Let $r$ be a voting rule and $\P$ a partial voting profile. 
\begin{itemize}
\item  The set $\NW(r,\P)$ of the \emph{necessary winners} with respect to $r$ and $\P$ is the intersection of the sets $\text{W}(r,\T)$, where $\T$ varies over all completions of $\P$. Thus, a candidate $c$ is a \emph{necessary winner} with respect to $r$ and $P$, if $c$ is a winner in $\text{W}(r,\T)$ for every completion $\T$ of $\P$.

\item  The set $\PW(r,\P)$ of the \emph{possible winners} with respect to $r$ and $\P$ is the union of the sets $\text{W}(r,\T)$, where $\T$ varies over all completions of $\P$.  Thus, a candidate $c$ is a \emph{possible winner} with respect to $r$ and $\P$, if $c$ is a winner in $\text{W}(r,\T)$ for at least one completion $\T$ of $\P$.
\end{itemize}

The notions of \emph{necessary unique winners} and \emph{possible unique winners} are defined in analogous manner. The preceding notions were introduced by Konczak and Lang~\cite{konczak2005voting}. Through a sequence of subsequent investigations by  Xia and Conitzer~\cite{DBLP:journals/jair/XiaC11}, Betzler and Dorn~\cite{DBLP:journals/jcss/BetzlerD10}, and Baumeister and Rothe~\cite{DBLP:journals/ipl/BaumeisterR12}, the following classification of the complexity of the necessary and the possible winners  for \emph{all} pure positional scoring rules was established.

\begin{theorem} \label{class-thm} {\rm[Classification Theorem]}
The following statements hold.
\begin{itemize}
\item If $r$ is a  pure positional scoring rule,  there is a polynomial-time algorithm that, given a partial voting profile $\P$,  returns the set $\NW(r,\P)$ of necessary winners. 
\item If $r$ is the \plurality rule or the \veto rule,  there is a polynomial-time algorithm that, given a partial voting profile $\P$,  returns the set $\PW(r,\P)$ of possible winners.  For all other pure positional scoring rules, the following problem is NP-complete: given a partial voting profile $\P$ and a candidate $c$, is $c$ a possible winner w.r.t.\ $r$ and $\P$?
\end{itemize}
Furthermore, the same classification holds for necessary unique winners and possible unique winners.
\end{theorem}

In the preceding theorem, the input partial voting profiles  consist of arbitrary partial orders.  There has been a growing body of work concerning the  complexity of the possible winners when the partial voting profiles are restricted to special types of partial orders. The main motivation for pursuing this line of investigation is to determine whether or not  the complexity of the possible winners drops from NP-complete to polynomial time w.r.t.\ some scoring rules (other than \plurality and \veto), if the input voting profiles consist of restricted partial orders that also arise naturally in real-life settings. We now describe two types of restricted partial orders and state  relevant results.

\begin{definition} \label{restr-order-defn}
Let $\succ$ be a partial order on a set $C$.
\begin{itemize}

  \item We say that $\succ $ is a \emph{partitioned preference} if $C$ can be partitioned into disjoint subsets $A_1, \ldots, A_q$ so that the following hold:
    
    (a) for all $i < j \leq  q$, if $c \in A_i$ and $c' \in  A_j$,  then $c \succ c'$;
    
    (b)  for each $i \leq q$, the elements in $A_i$ are incomparable under $\succ$, that is to say,    $a \nsucc b$ and $b \nsucc a$ hold, for all $a, b \in A_i$.
    \item We say that $\succ$ is a \emph{partial chain} if it consists of  a linear order on a non-empty subset $C'$ of $C$.
    
        \end{itemize}
        
        \end{definition}
        
Partitioned preferences relax the notion of a total order by requiring that there is a total order between sets of incomparable elements. As pointed out by Kenig \cite{DBLP:conf/atal/Kenig19}, partitioned preference ``were shown to be common in many real-life datasets, and have been used for learning statistical models on full and partial rankings.'' Furthermore, partitioned preferences contain \emph{doubly-truncated ballots} as a special case, where there is a complete ranking of  top elements, a complete  ranking of bottom elements, and all remaining elements between the top and the bottom elements are pairwise incomparable. This models, for example, the setting in which a voter has complete rankings  of some top candidates and of some bottom candidates, but is indifferent about the remaining candidates in the "middle". Partial chains arise in settings where there is a large number of candidates, but a voter has knowledge of only a subset of them. For example, a voter may have a complete ranking of movies that the voter has seen, but, of course, does not know how to compare these movies with movies that the voter has not seen. Partial chains also model the setting of an election in which one or more candidates enter the race late, and so a voter has a complete ranking of the original candidates but does not know yet how to rank the new candidates who entered the race late.

Let $r$ be a pure positional scoring rule.  We write \PWPV$(r)$ to denote the restriction of the possible winners problem w.r.t.\ $r$  to partial voting profiles consisting of partitioned preferences. More formally,  \PWPV$(r)$ is the following decision problem:   given a partial voting profile $\P$ consisting of partitioned preferences and a candidate $c$, is $c$ a possible winner w.r.t.\ $r$ and $\P$? Similarly, we write \PWPC$(r)$ to denote for the restriction of the possible winners problem w.r.t.\ $r$ to partial voting profiles consisting of partial chains.

Kenig \cite{DBLP:conf/atal/Kenig19} established a nearly complete classification of the complexity of the \PWPV$(r)$ problem for pure positional scoring rules. In particular, if $r$ is the $2$-approval rule, then \PWPV$(r)$ is solvable in polynomial time. In fact, the tractability of $\PWPV(r)$ holds for all $2$-\emph{valued} rules, that it, positional scoring rules in which the scoring vectors contain just two distinct values. If $r$ is the Borda rule, however, then \PWPV$(r)$ is NP-complete. In fact,  results in \cite{betzler2011unweighted,davies2011complexity} imply that the possible winners problem w.r.t.\ the Borda rule is NP-complete, even when restricted to input partial voting profiles consisting of doubly truncated ballots. As regards partial chains, it was shown recently in \cite{chakraborty2020complexity} that the classification in theorem \ref{class-thm} does not change for the \PWPC$(r)$ problem. In other words, for every positional scoring rule $r$ other than \plurality and \veto, \PWPC$(r)$ is NP-complete. In particular, \PWPC$(r)$ is NP-complete if $r$ is the 2-approval rule or the Borda rule.

Our experimental evaluation will focus on the \plurality rule,  the $2$-approval rule, and the Borda rule. For this reason, we summarize the aforementioned complexity results concerning these rules in the following table (also listing the veto rule for completeness).

\begin{table}[ht]
\centering
\begin{tabular}{l|ccc|c}  
\toprule
Scoring Rule  & $\PW$ & \PWPV & \PWPC  & $\NW$ (all kinds)\\
\toprule 
Plurality \& Veto& \Ptime  & \Ptime &  \Ptime  &  \Ptime    \\
\midrule
$2$-approval & NP-complete   & \Ptime  &  NP-complete  & \Ptime  \\
\midrule
Borda & NP-complete   & NP-complete  &  NP-complete    & \ \Ptime  \\
\bottomrule 
\end{tabular}
\newline 
\caption{Complexity of the possible winners  (PW) and necessary winners  (NW) problems, and their restrictions to  partitioned preferences (\PWPV) and partial chains (\PWPC) w.r.t.\ \plurality, \veto, $2$-approval, and Borda rules.}
\label{tab:results}
\end{table}

\colorTwo
 
\section{Necessary Winners}
\label{sec:nw}

Xia and Conitzer~\cite{DBLP:journals/jair/XiaC11} presented a polynomial-time algorithm for determining whether {\em a particular candidate} $c$ is a necessary winner (NW) in an election that uses a positional scoring rule $r$, that is, whether $c \in \NW(r,\P)$.  We recall it in Algorithm~\ref{alg:nw}.  We will then present several performance optimizations that allow us to efficiently compute {\em the set} $\NW(r,\P)$ of necessary winners.  

For a partial order $P \in \P$ and a candidate $c \in C$, we let $\textsc{Up}_{P}(c)= \{c'\in C | c' \succeq_{P} c\}$ and   $\textsc{Down}_{P}(c) = \{c'\in C| c' \preceq_{P} c\}$. (Note that both $\textsc{Up}_{P}(c)$ and $\textsc{Down}_{P}(c)$ include $c$.) Further, for a pair of candidates $c$ and $w$ with  $c \succ_{P} w$, we write 
$\textsc{Block}_{P}(c,w)= \textsc{Down}_{P}(c) \cap \textsc{Up}_{P}(w)$ for the set of candidates ranked between $c$ and $w$, including $c$ and $w$.

\begin{algorithm}[h!]
	\caption{$\checkNW(c, \P, r)$}
	\begin{algorithmic}[1]
	\FOR {each partial order $P \in \P$}
	\FOR {each candidate $c' \in C$}
	\STATE compute \up($c'$) and \down($c'$)
	\ENDFOR 
	\ENDFOR
	\FOR {each $w \in C \setminus c$}
	\STATE Initialize $S(w)=S(c)=0$
	\FOR {each partial order $P \in \P$}
	\IF { $c \not \succ_{P} w$}
    \STATE $pos_c = (m - \left|\textsc{Down}_{P}(c)\right| + 1)$ is the lowest possible position for $c$
    \STATE $pos_w = \left|\textsc{Up}_{P}(w)\right|$ is the highest possible position for $w$
    \STATE  $S(c) = S(c) + r(m,pos_c)$ 
    \STATE  $S(w) = S(w) + r(m,pos_w)$ 
    \ELSIF{$c \succ_{P} w$}
    \STATE slide  $\textsc{Block}_{P}(c,w)$ between positions $\left| \textsc{Up}_{P}(w) \setminus \textsc{Down}_{P}(c) \right| +1$ and $m-\left| \textsc{Down}_{P}(c) \setminus \textsc{Up}_{P}(w) \right|$, find positions $pos_c$ and $pos_w$ that minimize $r(m,pos_c) - r(m,pos_w)$
    \STATE  $S(c) = S(c) + r(m,pos_c)$
    \STATE  $S(w) = S(w) + r(m,pos_w)$ 
	\ENDIF 
	\ENDFOR
	\IF {$S(w) > S(c)$}\label{nw:check}
	\STATE $c$ is not a necessary winner, return \FALSE
	\ENDIF 
	\ENDFOR 
	\STATE $c$ is a necessary winner, return \TRUE 
	\end{algorithmic}
	\label{alg:nw}
\end{algorithm}

Note that Algorithm~\ref{alg:nw} returns true if $c$ is a necessary winner, not only if it's a necessary \emph{unique} winner.  To return true only if $c$ is the necessary unique winner, line~\ref{nw:check} should be replaced by $S(w) \geq S(c)$.



We now present several performance optimizations that allow us to efficiently compute {\em the set} $\NW(r,\P)$ of necessary winners.   Our optimizations are of two kinds.  The first kind is based on reusing computation across candidates, and on heuristically re-ordering computation.  The second kind uses the structure of a given partial voting profile to optimize the computation of $\textsc{Up}_{P}(c)$ and  $\textsc{Down}_{P}(c)$. 

\paragraph{Reusing and reordering computation.}  A straight-forward way to use Algorithm~\ref{alg:nw} to compute $\NW(r,\P)$ is to execute it $m$ times, once for each candidate.  

To eliminate redundant computation, we first compute and record the $\textsc{Up}_{P}(c)$ and  $\textsc{Down}_{P}(c)$ of each $P$ and $c$ once. We will explain how to compute $\textsc{Up}_{P}(c)$ and  $\textsc{Down}_{P}(c)$ efficiently later in this section.  Additionally, we compute and record the best possible score of each candidate, $S_{max}(c) = \textsc{argmax}_{P \in \P}~r(m,\left|\textsc{Up}_{P}(c) \right|)$.

Next, we execute competitions between pairs of candidates $c$ and $w$, deliberately selecting only the promising candidates as $c$, and prioritizing strong opponents $w$.   Specifically, only the candidates that have the highest $S_{max}(c)$ can become necessary winners.  Further, we sort potential opponents in decreasing order of $S_{max}(w)$. 

\paragraph{Computing $\textsc{Up}_{P}(c)$ and  $\textsc{Down}_{P}(c)$} This part of the computation takes polynomial time, but the details of this computation are left unspecified by Xia and Conitzer.  In our implementation, we use BFS algorithm to compute these sets for all candidates of a given partial profile $P$, maintaining intermediate results in a priority queue. 

We also observe that the structure of $P$ can be used to make this computation more efficient in some common cases.  In particular, $\textsc{Up}_{P}(c)$ and  $\textsc{Down}_{P}(c)$ can be computed in $O(m)$ time for {\em linear forests} (i.e., posets in which every candidate has at most one parent and at most one child) and for {\em partitioned preferences}, where candidates are partitioned into $q$ sets $C = A_1 \cup \ldots \cup A_q$, and where $P$ provides a complete order over the sets but does not compare candidates within a set.  A common example of partitioned preferences are top-$k$ preferences, where the first $k$ sets are of size 1, and the final set is of size $m-k$. Alternatively, $\textsc{Up}_{P}(c)$ and  $\textsc{Down}_{P}(c)$ computation can be avoided altogether in these cases, since scores of $S(c)$ and $S(w)$ that minimize $S(c) - S(w)$ can be determined directly.

\paragraph{In summary,} while the optimizations described in this section do not reduce the asymptotic running time of the already polynomial $\NW(r,\P)$ computation in the general case, they are useful in practice, as we will demonstrate experimentally in Section~\ref{sec:exp:nw} .  As we explain in the next section, we use these and similar techniques to optimize  the performance of $\PW(r,\P)$, making this computation practically feasible.
\section{Possible Winners} 
\label{sec:possible}

\subsection{Computing PW for Plurality and Veto}
\label{sec:pw-betzler-dorn}

By Theorem \ref{class-thm}, for  Plurality  and Veto, there are  polynomial-time algorithms for telling if a given candidate is a possible winner. In fact, Betzler and Dorn \cite{DBLP:journals/jcss/BetzlerD10} gave such an algorithm for plurality  by efficiently transforming the detection of possible winners to a network flow problem with just two layers and with integral capacities along the edges of the network. We have implemented and optimized this algorithm by, among other things, eliminating obvious winners (candidates ranked first in over half of the partial orders in $\P$) and obvious losers (candidates ranked first in fewer than $1/m$ partial orders), thus reducing the size of the network.  A variant of this algorithm can be used to detect possible winners for veto.

\subsection{Reducing PW to ILP} \label{pw-ilp}
Again by Theorem \ref{class-thm}, for all positional scoring rules other than plurality and veto, detecting possible winners is a NP-complete problem. Here, we give a polynomial-time reduction of the Possible Winners problem to Integer Linear Programming (ILP) and, in fact, to 0-1 ILP. Let $r$ be a positional scoring rule and let ${\bf s}=(s_1,\ldots,s_m)$ be its scoring vector for $m$ candidates. Consider an input to the possible winners problem consisting of a set $C=\{c_1,\ldots,c_m\}$ of candidates, a partial voting profile ${\P}=(P_1,\ldots, P_n)$, and a distinguished candidate $c_w$ from $C$; the question is whether or not $c_w \in \PW(r,{\P})$.

\begin{itemize}
    \item For each $l$ with $1\leq l\leq n$ and each $i$ with $1 \leq i \leq m$, introduce $m$ binary variables $x^{l,i}_1, x^{l,i}_2, x^{l,i}_3, \dots, x^{l,i}_m$.  Intuitively,  we want to have $x^{l,i}_j = 1$ if  candidate $c_i$ has rank $j$ in a completion $T_l$ of $P_l$; otherwise, $x^{l,i}_j = 0$. Thus, the  rank of $c_i$ in   $T_l$
    is equal to $\sum_{p=1}^{m} p \cdot  x^{l,i}_{p}.$ 

    \item There are two constraints to ensure the validity of a completion $T_l$ of $P_l$ , namely, each candidate is assigned exactly one rank in $T_l$, and no two candidates  are assigned the same rank in $T_l$.
    
    \begin{align}
    &  \sum_{p=1}^{m} x^{l,i}_p = 1, \text{where}~1\leq l\leq n ~ \text{and}~  1\leq i\leq m 
                                     \label{eqn::genPWtoILP_cndCstr}  \\
    &  \sum_{c_i \in C} x^{l,i}_p =1, \text{where}~1\leq l\leq n ~ \text{and}~  1\leq p\leq m. 
                                     \label{eqn::genPWtoILP_posCstr}
    \end{align}
 
    \item If a candidate  $c_i$ is ranked higher than a candidate $c_j$ in the partial order $P_l$, then $c_i$ has to also be ranked higher in a completion $T_l$ of $P_l$. This is ensured by introducing the following constraint for each such pair of candidates and each partial order.

    \begin{align}
        \sum_{p=1}^{m} p \big( x^{l,j}_{p} - x^{l,i}_{p} \big) > 0.
        \label{eqn::genPWtoILP_ppCstr}
    \end{align}

    \item Finally, to ensure that the distinguished candidate $c_w$ is a possible winner, we add, for each candidate $c_i\not = c_w$, the following constraint:
    \begin{align}
        \sum_{l=1}^{n} \sum_{p=1}^{m} s_{p} \cdot x^{l,i}_p  \leq  \sum_{l=1}^{n} \sum_{p=1}^{m} s_{p} \cdot x^{l,w}_p. 
        \label{eqn::genPWtoILP_pwCstr}
    \end{align}
\end{itemize}

Let  $\Sigma$ be the preceding ILP instance. Note that $\Sigma$  has $O(m^2n)$ binary variables and $O(m^2n)$ constraints. 

Note also that for the  case of possible unique winners, one has $s({\T}, c_i) < s({\T}, c_w).$ Thus, the only change needed is to replace the inequality in (\ref{eqn::genPWtoILP_pwCstr}) by a strict one.

We want to show that a 0-1 solution to $\Sigma$ exists if and only if candidate $c_w$ is a possible winner. We first prove a few facts about the reduction and then prove the desideratum. In the following, for a set $C$, we let $\Pi(C)$ denote the set of all total orders on $C$. We also let $\pi: \Pi(C) \times C \longrightarrow [1, \ldots, m]$ be the \emph{ranking} function that returns the rank of $c' \in C$ in a total order on $C$.

\begin{lemma}
\label{gen_fact_1}
Let $C = \{ c_1, \hdots, c_m \}$ and $\P = \{P_1, \hdots, P_n\}$ a set of partial votes. For each partial vote $P_l$, let $T_l  \in \Pi(C)$ be a total order that extends $P_l$. Consider the following
\begin{align}
    a_{p}^{l,i} = 
    \begin{cases}
    1,& \text{if } \pi(P_l, c_i) = p\\
    0,              & \text{ otherwise }
\end{cases}
\label{const_assign_a_gen}
\end{align}
The values $a_{p}^{l,i} , 1 \leq i \leq m, 1 \leq p \leq m$ have the following properties:
\begin{enumerate}
    \item $\sum_{p=1}^{m} a_{p}^{l,i} = 1$;
    \label{gen_val_prop_1}
    \item For $p = 1,  \hdots, m$, we have that  $\sum_{i=1}^{n} a^{l,i}_{p} = 1$.
     \label{gen_val_prop_2}
\end{enumerate}
\end{lemma}
\begin{proof}

These properties are a consequence of the way the ranking function $\pi: \Pi(C) \times C \longrightarrow [1, \hdots, m] $. For $1 \leq i \leq m$, in a completion $T_l$, of the partial vote $P_l$, each candidate is assigned a unique rank. Therefore, for a $T_l \in \Pi(C)$, for a fixed candidate $c \in C,$ there is exactly one $k \in \{1, \hdots, m \}$ such that $\pi(T_l, c) = k.$ By definition of the values $a^{l,i}_p$, only $a^{l,i}_k = 1;$ the others are $0.$ Thus, $\sum_{p=1}^m a^{l,i}_p = 1.$
\\
For property (\ref{gen_val_prop_2}), observe that $\pi$ assigns to each rank, an unique candidate. 
In the extension $T_l$, of a given vote $P_l,$ for a fixed  $k \in \{1, \hdots, m\}$,  there exists exactly one $c \in C$ such that $k \mapsfrom (T_l, c).$ 
Therefore, for $p = 1, \hdots, m$, we have that  $\larger{\sum_{i=1}^{i} a^{l,i}_p = 1}.$ 
\end{proof}

\begin{lemma}
\label{fact_val_prop}
Suppose $b_{p}^{l,i} \ 1 \leq i \leq n, 1 \leq p \leq m$ are non negative integers such that
\begin{enumerate}
    \item $b_{p}^{l,i} \in \{0,1\}$
    \label{val_prop_1}
   \item $\sum_{p=1}^{m} b_{p}^{l,i} = 1$
   \label{val_prop_2}
    \item For $p = 1, \hdots, m,$ we have $\sum_{i=1}^{n} b^{l,i}_{p} = 1$
    \label{val_prop_3}
\end{enumerate}
Let $\pi: \Pi(C) \times C \longrightarrow [1, \hdots, m]$ such that $\pi(T_l, c_i) = k_i$ if and only if $b_{k_i}^{l,i} = 1$. This induces a total order on $C.$
\end{lemma}
\begin{proof}
Fix a $T_l \in \Pi(C)$. Conditions (\ref{val_prop_1}) and (\ref{val_prop_2}) ensure that each $c_i \in C$ is mapped to exactly one $p \in \{1, \hdots, m\}.$ Conditions (\ref{val_prop_1}) and (\ref{val_prop_3}) ensure that for each $p \in \{1,  \hdots, m\},$ some $c_i \in C$ is mapped to it. Thus, $\pi$ is a one-to-one-correspondence.
\end{proof}
Observe that when $a^{l,i}_p$'s are defined as in Equation~(\ref{const_assign_a_gen}), 
$$\pi(T_l, c_i) = \sum_{p=1}^{m} p \cdot a^{l,i}_p .$$
Furthermore, for a scoring vector $\boldsymbol{s} = (s_1,\hdots, s_m),$ the score of $c_i$ in $T_l$ is
$$s(T_l, c_i) = \sum_{p=1}^{m} s_{p} \cdot a^{l,i}_p .$$
\begin{lemma}
\label{gen_fact_3}
Let $P_l$ be a partial order on $C$ and $T_l$ be a total order on $C.$ Suppose that we have values
$a_{p}^{l,i}$ as defined as in Equation~(\ref{const_assign_a_gen}). For all $c_i > c_j$ in $P_l$ we have the following inequalities for $T_l$
$$
    \larger{ \sum_{p=1}^{m}} p(a^{l,j}_p - a^{l,i}_p) \ \larger{ >} \ 0
$$
if and only if $P_l \hookrightarrow T_l.$
\end{lemma}
\begin{proof}
Fix a $P_l.$ Suppose, $c_i \succ c_j$ in $P_l.$ Further, suppose, that in a total order $T_l,$ the following holds $\larger{ \sum_{p=1}^{m}} p(a^{l,j}_p - a^{l,i}_p) \ \larger{ >} \ 0.$ This implies that   $\sum_{p=1}^{m} p \cdot a^{l,j}_p  - \sum_{p=1}^{m} p \cdot a^{l,i}_p \ \larger{ >} \ 0$, \ie $\pi(T_l, c_j) - \pi(T_l, c_i) > 0$, or $\pi(T_l, c_j) > \pi(T_l, c_i)$. This means that $c_i > c_j$ in $T_l.$ But this is also the case in $P_l.$ Therefore, $P_l \hookrightarrow T_l.$
\\
Let $P_l \hookrightarrow T_l.$ By definition, for all $c_i \succ c_j$ in $P_l,$ we have $c_i > c_j$ in $T_l$. Therefore, $\pi(T_l, c_j) > \pi(T_l, c_i)$, i.e., $\pi(T_l, c_j) - \pi(T_l, c_i) > 0$. This is nothing but $\sum_{p=1}^{m} p ( a^{l,j}_p  - a^{l,i}_p) \ \larger{ >} \ 0$.
\end{proof}

Consider a profile $\T = (T_1, \hdots, T_n)$ and the scoring vector $\boldsymbol{s} = (s_1, \hdots, s_m).$ The total score of a $c_i \in C$, by definition, is $\sum_{l=1}^{n} s(T_l, c_i)$ which is
$$
\label{total_score_profile}
\larger{ \sum_{l=1}^{n}} \sum_{p=1}^{m} s_{p} \cdot a^{l,i}_p  
$$
The above along with the definition of a possible winner makes the following fact quite obvious.
\begin{lemma}
\label{gen_PW_constraints}
Let $(T_1, \hdots, T_n )$ be a profile as above and let $c_w$ be a fixed candidate. Then the following statements are equivalent.
\begin{enumerate}
    \item $c_w$ is a winner in $\T$ using the scoring rule $\boldsymbol{s} = (s_1, \hdots, s_m)$.
    \item For every candidate $c_i \neq c_w$ we have that,
    $$
         \sum_{l=1}^{n} \sum_{p=1}^{m} s_{p} \cdot a^{l,i}_p \leq  \sum_{l=1}^{n} \sum_{p=1}^{m} s_{p} \cdot a^{l,w}_p
    $$
\end{enumerate}
\end{lemma}

Now we will prove the main theorem.     
      
\begin{theorem}
\label{thm:gen_PWinner}
The following statements are equivalent.
\begin{enumerate}
    \item Candidate $c_w$ is a possible winner w.r.t.\ the rule $r$ and the partial profile ${\P}$.
    \label{gen_PWinnerThm_1}
    \item The system $\Sigma$ has a 0-1 solution.
    \label{gen_PWinnerThm_2}
\end{enumerate}
\end{theorem}
\begin{proof}
$(\ref{gen_PWinnerThm_1} \implies \ref{gen_PWinnerThm_2})$ Assume that a partial order
$P = \{P_1, \hdots, P_n \} \hookrightarrow T = \{T_1, \hdots, T_n \}$ such that $c_w$ is a possible winner.
Set 
\begin{align}
    a_{p}^{l,i} = 
    \begin{cases}
    1,& \text{if } \pi(T_l, c_i) = p\\
    0,              & \text{ otherwise }
\end{cases}
\end{align}
We claim that  the assignment $x^{l,i}_p \mapsto a^{l,i}_p$ satisfy all the equations of the system $\Sigma.$
\\
Indeed, from Lemma \ref{gen_fact_1}, we know that this satisfies the following constraints
\begin{enumerate}
    \item $\sum_{p=1}^{m} a_{p}^{l,i} = 1$
    \item For $p = 1,  \hdots, m,$ we have $\sum_{i=1}^{n} a^{l,i}_{p} = 1$
\end{enumerate}
Since $\P \hookrightarrow \T$, by definition, $P_l \hookrightarrow T_l$. By Lemma  \ref{gen_fact_3}, the constraints
$\larger{ \sum_{p=1}^{m}} p(x^{l,j}_p - x^{l,i}_p) \ \larger{ >} \ 0$ are satisfied.
Since $c_w$ is a possible winner in $P$, by Lemma \ref{gen_PW_constraints}, the constraints
$$
         \sum_{l=1}^{n} \sum_{p=1}^{m} s_{ p} \cdot x^{l,i}_p \leq  \sum_{l=1}^{n} \sum_{p=1}^{m} s_{ p} \cdot x^{l,w}_p
    $$
     are satisfied.
\\
$(\ref{gen_PWinnerThm_2} \implies  \ref{gen_PWinnerThm_1})$Assume that the system $\Sigma$ has the integer solution $a^{l,i}_p \ (1 \leq l \leq n, 1 \leq i \leq m, 1 \leq p \leq m )$. The each $a^{l,i}_p$ is either $0$ or $1$ by the first group of constraints. Furthermore, by Lemma \ref{fact_val_prop}, the constraints 
\begin{align}
    \text{For each candidate }c_i \in C & \nonumber\\
                                    & \sum_{p=1}^{m} x^{l,i}_p = 1\\
    \text{For each position }p=1,  \hdots, m & \nonumber \\ 
                                    & \sum_{i=1}^{m} x^{l,i}_p =1
\end{align}
ensure that each vote induces a total order $>_l$ on $C.$ Furthermore, the total order $T_l$ extends $P_l$ because of the constraints $\sum_{p=1}^{m} p ( a^{l,j}_p  - a^{l,i}_p) \ \larger{ >} \ 0.$ Finally, since the constraints 
$$
         \sum_{l=1}^{n} \sum_{p=1}^{m} s_{ p} \cdot x^{l,i}_p \leq  \sum_{l=1}^{n} \sum_{p=1}^{m} s_{ p} \cdot x^{l,w}_p
    $$
   are satisfied, Lemma  \ref{gen_PW_constraints} implies that $c_w$ is a possible winner.
\end{proof}
We illustrate the preceding result in the following concrete cases.

\begin{description}
\item [Borda rule]
The scoring vector for Borda is $(m-1,m-2,\ldots,0)$. Thus,
 the constraints in (\ref{eqn::genPWtoILP_pwCstr}) become

\begin{align}
   \sum_{l=1}^{n} \sum_{p=1}^{m} (m-p) \cdot x^{l,i}_p  \leq  \sum_{l=1}^{n} \sum_{p=1}^{m} (m-p) \cdot x^{l,w}_p. 
   \label{eqn::Borda-constraint}
\end{align}


\eat{Note that for the PuW case, the inequality in the last constrain is replaced by a strict one.}

\item [$t$-approval]
The scoring vector for $t$-approval assigns a score of $1$ to each of the first $t$ ranked candidates, and $0$ to the remaining ones. Thus, the constraints in  (\ref{eqn::genPWtoILP_pwCstr}) become
\begin{align}
   \sum_{l=1}^{n} \sum_{p=1}^{t}  x^{l,i}_p  \leq  \sum_{l=1}^{n} \sum_{p=1}^{t}   x^{l,w}_p.
   \end{align}
\end{description}
\eat{
Instead of introducing the variables $\{ x^{l, i}_p \}_{p=1}^{m}$ as in the general case, here, one needs $\{x^{l,i}_p\}_{p=1}^{k}$. The resulting  system has $O(mnk)=O(mn)$ variables and $O(m^2n)$ constrains.}

\subsection{Checking a Possible Winner}\label{sec:possible:ilp}
Determining whether or not $c_w \in \PW(r,{\P})$ using our methodology involves the following two main steps.
\begin{enumerate}
    \item Construct the ILP model. Constraints (\ref{eqn::genPWtoILP_cndCstr}) and (\ref{eqn::genPWtoILP_posCstr}) depend only on $m$ and $n$, whereas, constraints (\ref{eqn::genPWtoILP_ppCstr}) and (\ref{eqn::genPWtoILP_pwCstr}) depend additionally on $c_w$ and on the partial profile.
    \item Solve the ILP model.
\end{enumerate}
Fix the values for $m$ and $n$. One  creates a \emph{partial} model for the corresponding $(m,n)$ with only constraints (\ref{eqn::genPWtoILP_cndCstr}) and (\ref{eqn::genPWtoILP_posCstr}). This is called pre-processing.  To save time, pre-processed models can be reused when the  candidate $c_w$, the partial profile $\P$, or both change. 
To solve a specific problem, one loads the appropriate pre-processed model, and updates it by adding constraints (\ref{eqn::genPWtoILP_ppCstr}) and (\ref{eqn::genPWtoILP_pwCstr}) before solving it.

\subsection{Three-phase Computation of the Set of Possible Winners}\label{sec:possible:threephase}

A straight-forward way to compute the set of possible winners $\PW(r,\P)$ is to execute the computation described in Section~\ref{sec:possible:ilp} above $m$ times, once for each candidate.   We now describe a more efficient method that uses pruning and early termination techniques, and heuristics to quickly identify clear possible winners.  This method involves three phases:

\begin{enumerate}
    \item Use $\NW(r,\P)$ to identify a subset of possible winners $C^{1}_{pw}$, and to prune clear non-winners $C^{1}_{lsr}$.  
    Pass the remaining $C^{1} = C \setminus (C^{1}_{pw} \cup  C^{1}_{lsr})$ to the next phase.
    \item Use a heuristic to construct a completion in which $c \in C^{1}$ is a winner.  Add all candidates for which such a completion is found to $C^{2}_{pw}$, and pass the remaining $C^{2} = C^{1} \setminus C^{2}_{pw}$ to the next phase.
   \item Invoke the subroutine described in Section~\ref{sec:possible:ilp} to check a possible winner for each $c \in C^{2}$ using an ILP solver.  Add all identified possible winners to $C^{3}_{pw}$.  
\end{enumerate}

The final set of possible winners is $C^{1}_{pw} \cup C^{2}_{pw} \cup C^{3}_{pw}$.


\paragraph{Phase 1: Using the Necessary Winner algorithm.}   Let us denote by $S_{total}(r,m)$ the sum of scores of all candidates in some total voting profile.
We will execute  $\NW(r,\P)$ to compute the set of necessary winners, which are also possible winners.  Recall that as part of the  $\NW(r,\P)$ computation, we compute and record, for all $c \in C$, the best possible score $S_{max}(c) = \textsc{argmax}_{P \in \P}~r(m,\left|\textsc{Up}_{P}(c) \right|)$.
We can immediately identify candidates whose $S_{max}(c)$ is highest as possible winners, and add them to $C^{1}_{pw}$. Further, if $S_{max}(c) > \frac{1}{2} S_{total}(r,m)$, then $c$ is also a possible winner, and is added to $C^{1}_{pw}$.    

On the other hand, if $S_{max}(c) < \frac{1}{m} S_{total}(r,m)$ then $c$ is not a possible winner (by pigeonhole principle), and it can be pruned.  Further, consider the step in $\NW(r,\P)$ where we execute competitions between pairs of candidates $c$ and $w$.  As we compute $S(w)$ and $S(c)$, we may observe that $S(w)- S(c) < 0$.  This allows us to prune $w$ as a non-winner, adding it to $C^{1}_{lsr}$.
 


\paragraph{Phase 2: Constructing a completion.}  Next, given a candidate $c$, we consider ${\P}=(P_1,\ldots,P_n)$ and heuristically attempt to create  a total voting profile  ${\T}=(T_1,\ldots,T_n)$ that completes $\P$
and in  which $c$ is the winner.  If such a $\T$ is found, then $c$ is added to $C^{2}_{pw}$.  
To construct $\T$, we complete each partial vote $P \in {\P}$ independently, as follows:

    (1) For a given $P$, place $c$ at the worst possible rank in which it achieves its best possible score.  The reason for this is to minimize the scores of the items in $\textsc{Up}_{P}(c) \setminus {c}$.
    
    (2) Place the remaining candidates from $P$ into $T$.  If multiple placements are possible, chose one that increases the score of the currently highest-scoring candidates the least.
    
    (3) Keep a list of candidates other than $c$ that are the possible winners so far.  In subsequent completions,  place these candidates as low as possible, minimizing their score.

\paragraph{In summary,} we described a reduction of the problem of checking whether a candidate $c$ is a possible winner to an ILP, and proposed a three-phase computation that limits the number of times the ILP solver is invoked for a set of candidates $C$.  We will show experimentally in Section~\ref{sec:exp:pw} that the proposed techniques can be used to compute the set of possible winners in realistic scenarios. 

\section{The Repeated Selection Model for Poset Generation}\label{sec:rsm}

In this section we introduce a novel generative model for partially ordered sets, called the Repeated Selection Model, or RSM for short.   It  includes earlier generative models of partial orders as special cases via a suitable choice of parameters.  We regard RSM as being a model of independent interest, and we also use it here as part of our experimental evaluation, described in Section~\ref{sec:exp}.  To start, we introduce the Repeated Insertion Model (RIM) that is used for generating total orders in Section~\ref{sec:rsm:rim}.  We then describe our novel RSM model in Section~\ref{sec:rsm:rsm}.

\subsection{Preliminaries: The Repeated Insertion Model (RIM)}
\label{sec:rsm:rim}

In this section we represent total orders using rankings, that is, ordered lists of items indexed by position.  We will use $\bsigma$, $\btau$, and so on to denote rankings. We will use $\bsigma(i)$ to refer to an item at position $i$ in  $\bsigma$, and we will use $\bsigma^{-1}(a)$ to denote the position of element $a$ in $\bsigma$.  When describing iterative algorithms, for convenience of presentation we will denote by $\bsigma_i$ the value of $\bsigma$ at step $i$.

The Repeated Insertion Model (RIM) is a generative  model that defines a probability distribution over rankings due to Doignon \etal~\cite{Doignon2004}.  This distribution, denoted by $\RIM(\bsigma, \Pi)$, is parameterized by a reference ranking $\bsigma$ and a function $\Pi$, where $\Pi(i, j)$ is the probability of inserting $\bsigma(i)$ at position $j$.  Here, $\Pi$ is a matrix  where each row corresponds to a valid probability distribution (i.e., the values in a row sum up to one). Algorithm~\ref{alg:rim} presents the RIM sampling procedure. It starts with an empty ranking $\btau$, inserts items in the order of $\bsigma$, and puts item $\bsigma(i)$ at $j^{th}$ position of the currently incomplete $\btau$ with probability $\Pi(i, j)$.  The algorithm terminates after $m$ iterations, and outputs $\btau$, a total order over the items drawn from $\bsigma$.

\begin{algorithm}[h!]
	\caption{$\RIM(\bsigma, \Pi)$}
	\begin{algorithmic}[1]
		\STATE Initialize an empty ranking $\btau
		= \ranking{}$.
		\FOR {$i = 1, \ldots, m$}
		\STATE Select a random position $j\in [1, i]$ with a probability $\Pi(i,j)$
		\STATE Insert $\bsigma(i)$ into $\btau$ at position $j$
		\ENDFOR
		\RETURN  $\btau$
	\end{algorithmic}
	\label{alg:rim}
\end{algorithm}

\begin{example}
	$\RIM(\ranking{a, b, c}, \Pi)$ generates $\btau {=} \ranking{b, c, a}$ as follows.
	\begin{itemize}
	\item Initialize an empty ranking $\btau_0 {=} \ranking{}$. 
	\item At step 1, $\btau_1 {=} \ranking{a}$ by inserting $a$ into $\btau_0$ with probability $\Pi(1,1) {=} 1$. 
	\item At step 2, $\btau_2 {=} \ranking{b, a}$ by inserting $b$ into $\btau_1$ at position 1 with probability $\Pi(2,1)$. Note that $b$ is put before $a$ since $b \succ_{\btau} a$.
	\item At step 3, $\btau {=} \ranking{b, c, a}$ by inserting $c$ into $\btau_2$ at position 2 with probability $\Pi(3,2)$. 
	\end{itemize}
	The overall probability of sampling $\btau$ is $\Pr(\btau \mid \ranking{a,b,c}, \Pi) {=} \Pi(1,1) \cdot \Pi(2,1) \cdot \Pi(3,2)$.  Note that this particular sequence of steps is the only way to sample $\ranking{b, c, a}$ from $\RIM(\ranking{a, b, c}, \Pi)$.
		\label{rim:example}
\end{example}

The Mallows model~\cite{Mallows1957}, $\mallows(\bsigma, \phi), \phi \in (0, 1]$, is a special case of RIM.
As a popular preference model, it defines a distribution of rankings that is analogous to the Gaussian distribution: the ranking $\bsigma$ is at the center, and rankings closer to $\bsigma$ have higher probabilities.  Specifically, the probability of a ranking $\btau$ is given by:

\begin{equation}
\Pr(\btau|\mallows(\bsigma, \phi)) = \dfrac{\phi^{\dist(\bsigma, \btau)}}{1\cdot(1+\phi)\cdot(1+\phi+\phi^2)\ldots(1+\ldots+\phi^{m-1})}
\label{eq:mallows}
\end{equation}

Here, $\dist(\bsigma, \btau)$ is the Kendall-tau distance between $\bsigma$ and $\btau$: $\dist(\bsigma, \btau) = |{(a, a') \mid a \succ_{\bsigma} a', a' \succ_{\btau} a}|$, that is the number of preference pairs $(a, a')$ that appear in the opposite relative order in $\bsigma$ and $\btau$.  The expression in the denominator of Equation~\ref{eq:mallows} is the normalization constant, which we will find convenient to denote $Z_{\phi,m}$. When $\phi \rightarrow 0$, the probability mass is concentrated around the reference ranking $\bsigma$; when $\phi=1$, all rankings have the same probability, that is, $\mallows(\bsigma, 1)$ is the uniform distribution over rankings.

As was shown in~\cite{Doignon2004}, $\RIM(\bsigma, \Pi)$ is precisely $\mallows(\bsigma, \phi)$ when $\Pi(i, j) = \frac{\phi^{i-j}}{1+\phi+...+\phi^{i-1}}$.  That is, the Mallows model is a special case of RIM, and so RIM can be used as an efficient sampler for Mallows.



 

\subsection{The Repeated Selection Model (RSM)}\label{sec:rsm:rsm}

The Repeated Selection Model (RSM) is a generative model that defines a probability distribution over posets.  
Intuitively, in this model we iteratively select a random item and randomly choose whether it succeeds each of the remaining items. More formally, an instance of this distribution, denoted $\RSM(\bsigma,\Pi, p)$, is parameterized by a reference ranking $\bsigma$ of length $m$, a selection probability function $\Pi$, where $\Pi(i,j)$ is the probability of \emph{selecting} the $j^{th}$ item among the remaining items at step $i$, and a preference probability function $p:\{1,\dots,m-1\}\rightarrow[0,1]$ that determines the probability $p(i)$ that the $i^{th}$ selected item precedes (is preferred to) each of the remaining items. We view $\Pi$ as a matrix  where each row corresponds to a valid probability distribution (i.e., the values in a row sum up to one) and the $i-1$ rightmost entries in the $i^{th}$ row are zero.

Algorithm~\ref{alg:rsm} presents the RSM sampling procedure. Intuitively, in contrast to RIM (Algorithm~\ref{alg:rim}) that considers candidates one by one in the order of $\bsigma$ and inserts them into the output $\btau$, RSM iteratively selects, and removes, candidates one by one from among the remaining candidates in $\bsigma$ at each step.  Which candidate is selected at step $i$ is decided randomly, based on the probability distribution in the $i^{th}$ row of the selection probability matrix $\Pi$ (line~\ref{alg:rsm:pick} of Algorithm~\ref{alg:rsm}).  Furthermore, to generate posets rather than total orders, RSM uses the preference probability function $p$ to decide whether to add a particular preference pair to $\btau$; this decision is made independently for all considered pairs (lines~\ref{alg:rsm:prefstart}-\ref{alg:rsm:prefend} of Algorithm~\ref{alg:rsm}).  The probability that a candidate selected at step $i$ is preferred to each of the remaining candidates in $\bsigma$ is $p(i)$.

\begin{algorithm}[h!]
	\caption{$\RSM(\bsigma, p, \Pi)$}
	\begin{algorithmic}[1]
		\STATE Initialize an empty poset $\btau = \ranking{}$.
		\FOR {$i = 1, \ldots, m-1$}
		\STATE Select a random position $j\in [1, m]$ with a probability $\Pi(i,j)$
		\STATE Select candidate $c=\bsigma(j)$  \label{alg:rsm:pick}
		\STATE Remove $c$ from $\bsigma$, which now contains $m-i$ candidates
		 \FOR {$k = 1, \ldots, m-i$} \label{alg:rsm:prefstart}
		 \STATE Add the pair $c \succ \bsigma(k)$ to $\btau$ with probability $p(i)$ (or leave it out with probability $1-p(i)$)
		 \ENDFOR \label{alg:rsm:prefend}
        \ENDFOR
		\RETURN the transitive closure of $\btau$
	\end{algorithmic}
	\label{alg:rsm}
\end{algorithm}

\begin{example}
	$\RSM(\ranking{a, b, c}, \Pi, p)$ can generate $\btau {=} \ranking{b \succ a, a \succ c, b \succ c}$ as follows:
	\begin{itemize}
	\item Initialize an empty poset $\btau{=}\ranking{}$. 
	
	\item At step $i=1$, select $b$ with probability $\Pi(1,2)$ and remove it from $\bsigma$, setting $\bsigma_1{=}\ranking{a, c}$.  
	Then, add the pair $b \succ a$ to $\btau$ with probability $p(1)$, and \emph{do not add} the pair $b \succ c$ to $\btau$ with probability $1-p(1)$.  
	
	\item At step $i=2$, select $a$ with probability $\Pi(2,1)$ and remove it from $\bsigma$, setting $\bsigma_2{=}\ranking{c}$.  
	Finally, add the pair $a \succ c$ to $\btau$ with probability $p(2)$.
	
	\item Take the transitive closure of $\btau$ and return $\ranking{b \succ a, a \succ c, b \succ c}$.
		\end{itemize}
		
	The probability of sampling $\btau$ in this way is
	$\Pi(1,2) \cdot p(1) \cdot (1- p(1)) \cdot \Pi(2,1) \cdot p(2)$.   
	\label{rsm:example}
\end{example}
		
Note that the same $\btau$ can be generated by $\RSM(\ranking{a, b, c}, \Pi, p)$ using a different sequence of steps, thus yielding a different probability.  In our example there is one other way to derive $\btau$: at step $i=1$,  add $b \succ c$ to $\btau$ with probability $p(1)$. This happens with the probability $\Pi(1,2) \cdot p(1) \cdot p(1) \cdot \Pi(2,1) \cdot p(2)$.  
These are the only two possible derivations of $\btau$ in our example.  These  yield the total probability
\[\Pr(b \succ a \succ c \mid \ranking{a,b,c}, \Pi, p) = \Pi(1,2) \cdot p(1) \cdot \Pi(2,1) \cdot p(2)\,.\]
In the general case, however, it is not clear whether this  probability can be computed efficiently. In particular, 
the probability of a poset may be due to all the linear extensions of the poset.

Importantly, RSM includes several generative models of partial orders as special cases via a suitable choice of parameters. For example,  $p= (1,...,1,0,...,0)$ with $k$ ones  will generate a top-truncated partial order, whereas $p=(0,...,0,1,...,1)$ with $k$ ones will generate a partial chain over a subset of $k$ items.  Moreover, a uniform $p$ gives rise to the generative model referred to as \emph{Method 1} of Gehrlein~\cite{gehrlein1986methods}. Figure~\ref{fig:RSM_eval_gehrlein} compares these probability distributions empirically.  Finally, as we show below, the Mallows model is also a special case of RSM.


\begin{theorem} \label{thm:RSM_Mallows}
For a given $\phi \in (0,1]$, and for $p(i)=1$ for all $i$, we have that $\RSM(\bsigma,\Pi, p)$ is precisely $\mallows(\bsigma, \phi)$ when  $\Pi(i,j) = \frac{\phi^{j-1}}{\sum^{m-i+1}_{k=1}\phi^{k-1}}$. 
\end{theorem}

\begin{proof}
First, observe that because $p(i)=1$ for all $i$, $\RSM(\bsigma,\Pi, p)$ will generate total orders. Further, observe that, although an arbitrary poset can be generated by RSM in multiple ways (as demonstrated by Example~\ref{rsm:example}), there is only one way to obtain a total order (ranking). 
We will show by induction on the number of candidates $m$ in $\bsigma$ that $\Pr(\btau \mid \RSM(\bsigma,\Pi,p)) = \Pr(\btau \mid MAL(\bsigma,\phi))$.  Recall from Equation~\ref{eq:mallows} that $\Pr(\btau \mid MAL(\bsigma,\phi)) = \phi^{\dist(\bsigma,\btau)}/Z_{\phi,m}$, where $Z_{\phi,m}$ is a normalization constant, and $\dist(\bsigma, \btau)$ is the Kendall-tau distance between $\bsigma$ and $\btau$: 
$\dist(\bsigma, \btau) = |{(a, a') \mid a \succ_{\bsigma} a', a' \succ_{\btau} a}|$, that is the number of preference pairs $(a, a')$ that appear in the opposite relative order in $\bsigma$ and $\btau$.  For notational convenience, we will denote by $\bsigma_{-a}$ a subranking of $\bsigma$ with item $a$ removed.  Further, we will denote by $\Pi_{-i,-j}$ a projection of the matrix $\Pi$ with the $i^{th}$ column and $j^{th}$ row removed.

 \vskip1em
\noindent\underline{Base case.} When $m = 1$, both RSM and Mallows generate a single ranking with probability 1.

    
    
    \vskip1em
 \noindent\underline{Inductive step.}  
    Suppose that RSM and Mallows assign the same probability to the subranking of some $\btau$ with the first element $\btau(1)$, denoted $\tau_1$, removed, and with $\bsigma$ and $\Pi$ adjusted accordingly:
    
    \begin{equation}
    \Pr(\btau_{-\tau_1} \mid \RSM(\bsigma_{-\tau_1},\Pi_{-1,-m},p)) = \Pr(\btau_{-\tau_1} \mid \mallows(\bsigma_{-\tau_1},\phi))  
    = \frac{\phi^{\dist(\bsigma_{-\tau_1}, {\btau_{-\tau_1} )}}} {Z_{\phi,m-1}}
    \label{eq:subr}
    \end{equation}
    
    

   
Let us now consider the ranking $\btau$ of length $m$, and observe that $dist(\btau,\bsigma)=dist(\btau_{-\tau_1},\bsigma_{-\tau_1}) + \bsigma^{-1}(\tau_1) -1$, where $\bsigma^{-1}(\tau_1)$ is the position of element $\tau_1$ in $\bsigma$.
%
%
Moreover, the probability to select $\tau_1$ at the first step of RSM is given by: 
\begin{equation}
  \Pi(1,\sigma^{-1}(\tau_1)) = \frac{\phi^{\sigma^{-1}(\tau_1)-1}}{\sum^{m}_{k=1}\phi^{k-1}} 
  \label{eq:pi}
\end{equation}


Combining Equations~\ref{eq:subr} and~\ref{eq:pi}, and recalling the expression for $Z_{\phi,m}$ from Equation~\ref{eq:mallows}, we obtain the following probability for $\btau$ :
    \begin{align*}
        \Pr(\btau  \mid \RSM(\bsigma,\Pi,p)) &= \Pi(1,\sigma^{-1}(\tau_1)) \times \Pr(\btau_{-\tau_1} \mid \RSM(\bsigma_{-\tau_1},\Pi_{-1,-m},p)) \\
        &= \frac{\phi^{\sigma^{-1}(\tau_1)-1}}{\sum^{m}_{k=1}\phi^{k-1}} \times \frac{\phi^{\dist(\bsigma_{-\tau_1}, {\btau_{-\tau_1} )}}} {Z_{\phi,m-1}} 
        = \dfrac{\phi^{\dist(\btau,\bsigma)}}{Z_{\phi,m}}
        = \Pr(\btau \mid \mallows(\bsigma,\phi)). 
    \end{align*}


The proof by induction concludes and the theorem is proven.
\end{proof}

\section{Experimental Evaluation}
\label{sec:exp}

All experiments were carried out on an Intel(R) Xeon(R) CPU E5-2680 v3 @ 2.50GHz, with 412 GB of RAM, 20 hyper-threaded cores running 2 threads per core, running Ubuntu 16.04.6 LTS. We used Python 3.5 for our implementation, and the solver Gurobi v8.1.1 ~\cite{gurobi_2019} for solving the ILP instances produced by the reduction from instances of the possible winner problem.  

\subsection{Experimental Datasets and Scoring Rules}
\label{sec:exp:setup}  

\paragraph{Real datasets} We used two real datasets in our experimental evaluation, \emph{travel} and \emph{dessert}. 

The Google Travel Review Ratings dataset (\emph{travel})~\cite{renjith2018evaluation} consists of average ratings (each between 1 and 5) issued by 5,456 users for up to 24 travel categories in Europe.  For each user, we create a set of preference pairs such that items in each pair have different ratings (no tied pairs). Items for which a user does not provide a rating are not included into that user's preferences.  Because preferences are derived from ratings issued by individual users, there cannot be any cycles in the set of preference pairs corresponding to a given user.

The \emph{dessert} dataset was collected by us.  It consists of user preferences over pairs of eight desserts, collected from  228 users, with up to 28 pairwise judgments per user. For each pair, users indicated their confidence in the preference using a sliding bar.  This enabled us to create several voting profiles based on this data, each corresponding to a particular confidence threshold.  With a high confidence threshold, we keep fewer pairs and obtain a sparse profile, and with a low confidence threshold, we keep more pairs and obtain a very dense profile.   Because preferences are collected pairwise, there can be cycles.  We check the set of preferences of each user and only keep those that are acyclic for the experiments in this paper.

\paragraph{Synthetic datasets} 
We use three different types of synthetic voting profiles, namely, \emph{partial chains}, \emph{partitioned preferences}, and \rsmm.  We now describe the data generation process for each.

Recall from Definition~\ref{restr-order-defn} in Section~\ref{sec:preliminaries} that a \emph{partial chain} on a set $C$ is a partial order on $C$ that consists of a linear order on a non-empty subset $C'$ of $C$.  Further, recall that a \emph{partitioned preference} on a set $C$ is a partial order on $C$ with the property that $C$ is partitioned into disjoint subsets $A_1, \ldots, A_q$ such that (a) every element from $A_i$ is preferred to every element from $A_j$, for $i < j \leq q$; and (b) the elements in each $A_i$ are pairwise incomparable. 

We are given the set of candidates $C = \{c_1, c_2, c_3, \dots, c_m\}$, and the number of voters $n$.  To generate a partial chains profile or a partitioned preferences profile, we start with a mixture of three Mallows models, each with $\phi=0.5$,  with a randomly chosen $\sigma$ of size $m$, and covering approximately $\frac{1}{3}$ of the voters, and generate a complete voting profile  $\T=(T_1,\ldots,T_n)$ of total orders on $C$.  Then, to generate a partial chains profile, for each total order $T_i$, choose $d \in [0,m-2]$ uniformly at random, and drop $d$ candidates from $T_i$ by selecting one uniformly at random from the remaining candidates over $d$ iterations.  To generate a partitioned preferences profile, for each $T_i$, choose the number $q$  of non-empty partitions  uniformly at random from the set $[2,m]$.  To partition $T_i$, select $q-1$ positions between 2 and $m$ uniformly at random without replacement, with each position corresponding to the start of a new partition.  Drop the order relations between candidates in the same partition.

To generate an \rsmm voting profile, we use a mixture of three RSMs (as described in Section~\ref{sec:rsm}), each covering $\frac{1}{3}$ of the voters, with selection probability $\Pi_{i,j}$ corresponding to the Mallows model ($\phi=0.5$, randomly chosen $\sigma$ of size $m$). For each of the three RSMs, we draw the preference probability $p(m)$ uniformly from $[0,1]$ for each $m$.
    


\paragraph{Scoring rules}
We evaluated the performance of our techniques for three positional scoring rules, namely, the \plurality~rule, the $2$-approval rule, and the Borda rule.  We chose these rules for two reasons. First, they are arguably among the most well known and extensively studied positional scoring rules.  Second, the plurality rule and the $2$-approval rule are prototypical examples of \emph{bounded-value} rules, that is, rules in which the scores are of  bounded size (in this case, the bound on the size is $2$), while the Borda rule is a prototypical example of an \emph{unbounded-value} rule, that is, the scores may grow beyond any fixed bound. Note that we also conducted experiments for the 
\veto~rule and found out that performance followed the same trends as those for \plurality. We remind the reader that the results about the complexity of the necessary winners and the possible winners with respect to the plurality rule, the $2$-approval rule, and the Borda rule are summarized in Table \ref{tab:results} in Section \ref{sec:preliminaries}.

\eat{
Performance for these rules is representative of performance for other positional scoring rules, and together they cover the complexity classes for possible winners (PW) and necessary winners (NW), and their restrictions to partitioned preferences and partial chains, as discussed in Section~\ref{sec:preliminaries} and summarized in Table~\ref{tab:results} in that section.  Performance for the \veto scoring rule followed the same trends as that for \plurality, and so we only include results for \plurality here.  \julia{Phokion, we don't do experiments on veto, only on plurality, because performance is similar for both. I noted that here, need to rephrase in Section~\ref{sec:preliminaries} slightly as well, right before the table.} 
}

\subsection{Validation of the Repeated Selection Model (RSM)}\label{sec:exp:rsm}  

In this section we compare the Repeated Selection Model (RSM) with {\em Method 1} and {\em Method 2} from Gehrlein~\cite{gehrlein1986methods}.  Our first comparison is of the  empirical distribution of poset density, defined as $d = \frac{D}{{m\choose 2}} = \frac{2D }{m(m-1)}$, where $m$ is the number of items, and $D$ is the total number of preference pairs in the partial voting profile $\P$. Figure~\ref{fig:RSM_eval_gehrlein} presents this comparison for 10 candidates and 50 voters. We observe that the RSM generates partial orders over a wider range of densities than either of the two methods from Gehrlein.

\begin{figure}[t!]
    \centering
    \includegraphics[scale=0.475]{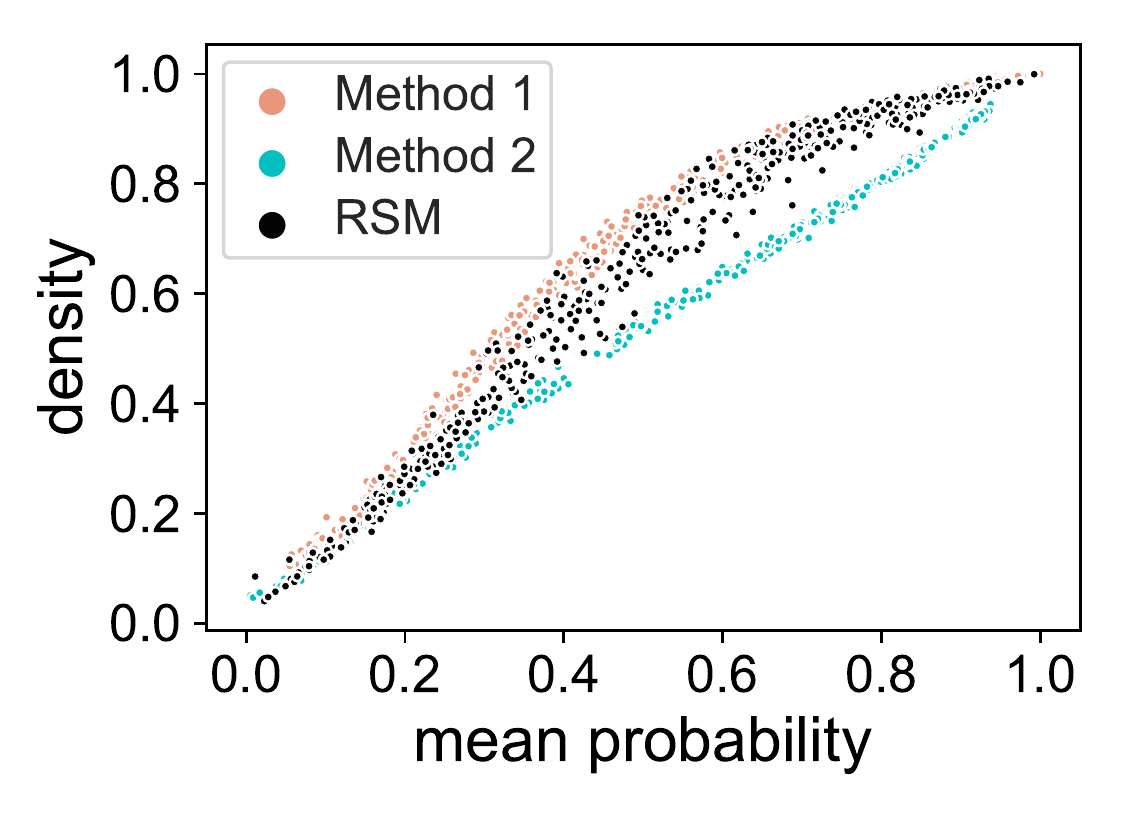}
  \caption{Density distribution of RSM, compared to two poset generation methods from Gehrlein~\cite{gehrlein1986methods}.}
  \label{fig:RSM_eval_gehrlein}
\end{figure}

\begin{table}[b!]
    \centering
    \begin{tabular}{c|c|c|c}
    \toprule
        \multirow{2}{*}{\emph{Dataset}} & \multicolumn{3}{|c}{$\mathcal{NLL}$}  \\
        \cline{2-4}
         &  \emph{RSM} & \emph{Method 1} & \emph{Method 2} \\
         \midrule\midrule
         \emph{Travel} & $\mathbf{9.5}$ & $10.4$  &$10.4$ \\  
         \midrule
         \emph{Dessert (Sparse)} & $\mathbf{12.5}$ & $12.9$ & $14.7$\\  
         \midrule
         \emph{Dessert (Dense)}& $\mathbf{18.7}$ & $19.1$ & $19.8$ \\ 
    \bottomrule
    \end{tabular}
    \caption{Comparison of goodness-of-fit of RSM (Section~\ref{sec:rsm}), and Methods 1 and 2 from Gehrlein~\cite{gehrlein1986methods}, on real-world datasets. $\mathcal{NLL}$ stands for negative log-likelihood, with lower values corresponding to better fit.}
    \label{tab:nll}
\end{table}

We also conducted an experiment to verify that RSM is sufficiently flexible to represent real partial voting profiles.  To do this, we extended the methods of Stoyanovich et al.~\cite{DBLP:conf/webdb/StoyanovichIP16} to fit a single RSM to {\em dessert} and {\em travel} datasets, and compared the goodness of fit to that of Method 1 and Method 2 from Gehrlein~\cite{gehrlein1986methods}. For each dataset, we compute the negative log-likelihood using the voters for whom we know both the real subranking $\btau$ and the synthetically generated subranking $\btau'$:
\[
\mathcal{NLL} = -\frac{1}{n}\sum^{n}_{i}{log(\Pr(\btau'_i\mid\btau_i))}
\]

Our results are summarized in Table~\ref{tab:nll}, and confirm that RSM fits these real datasets more closely than other methods, as quantified by negative log-likelihood ($\mathcal{NLL}$), with lower values corresponding to better fit. Note that all methods fit the \emph{travel} dataset better as compared to the \emph{dessert} dataset because the former contains partitioned preferences with missing candidates.


\subsection{Necessary Winners} 
\label{sec:exp:nw}

In this section, we evaluate the performance of an optimized version of the polynomial-time algorithm by Xia and Conitzer \cite{DBLP:journals/jair/XiaC11}, as described in Section \ref{sec:nw} for three positional scoring rules: \plurality, $2$-approval, and Borda. 

We start with experiments that demonstrate the impact of the  number of voters $n$ and the number of candidates $m$ on the running time of the optimized necessary winners algorithm described in Section~\ref{sec:nw}. In Figure~\ref{fig:nw_voters}, we set $m=100$, vary $n$ between 10 and 10,000 on a logarithmic scale, and show the running time for each of the rules \plurality,  $2$-approval, and Borda, and for each family of synthetic datasets as a box-and-whiskers plot.  We observe that the computation is efficient: RSM Mix is the most challenging, and completes in 10 seconds or less for 10,000 voters across all scoring rules. The running time increases linearly with $n$.

\begin{figure}[t!]
    \begin{subfigure}{0.3\textwidth}
    \centering
    \includegraphics[width=\linewidth]{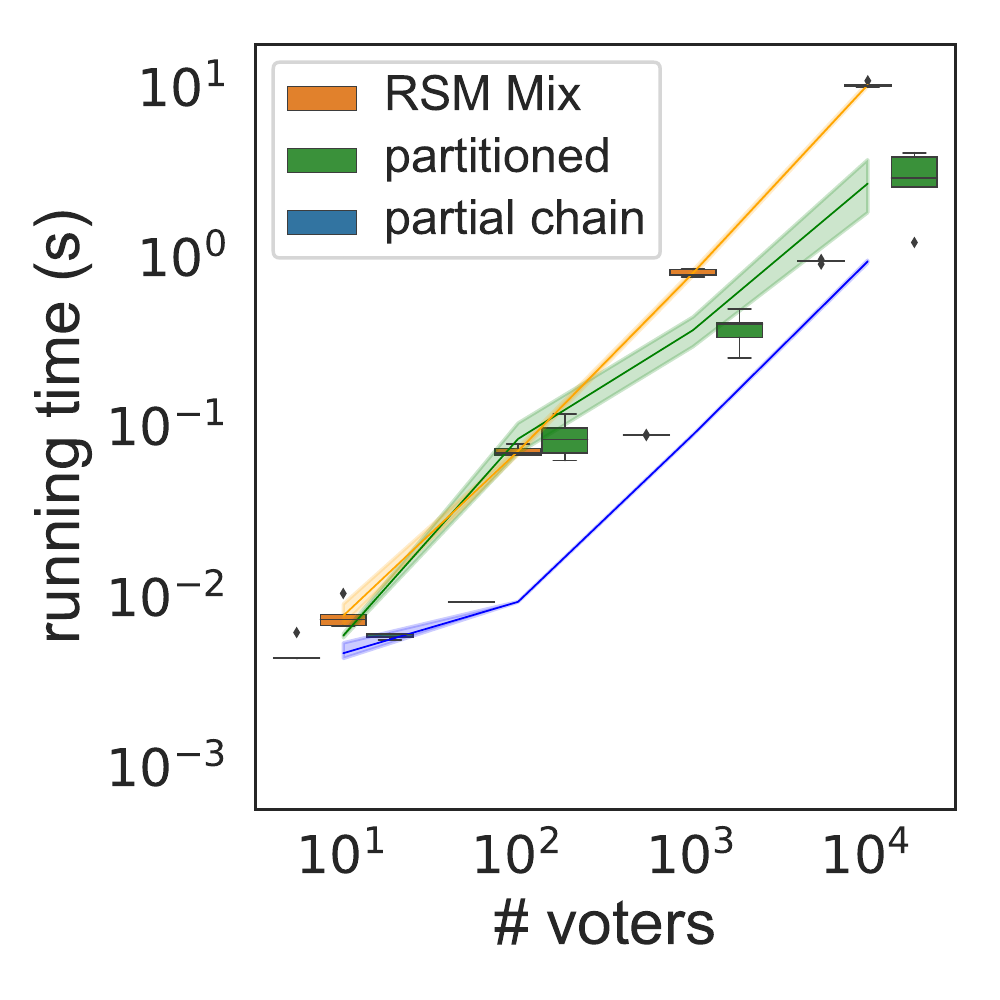}
    \caption{plurality}
    \label{fig:nw_voters_plurality}
    \end{subfigure}
    \begin{subfigure}{0.3\textwidth}
    \centering
    \includegraphics[width=\linewidth]{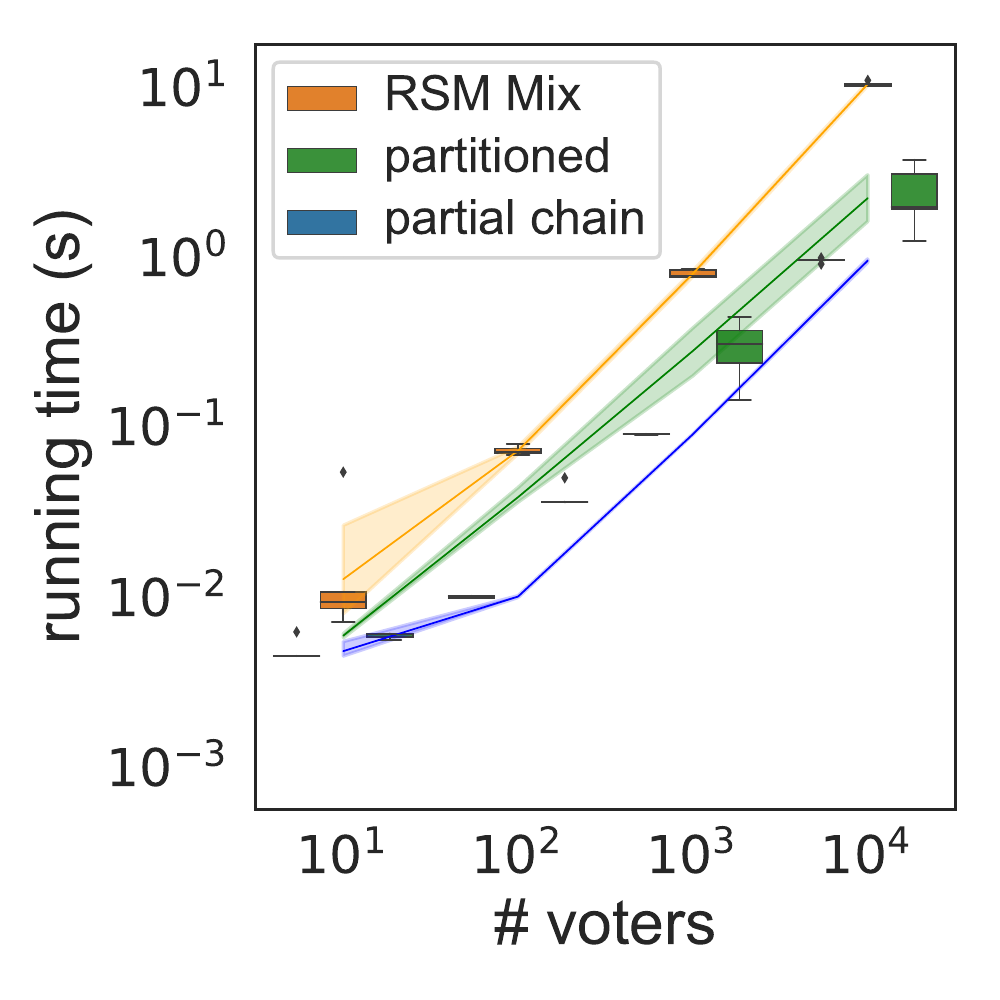}
    \caption{$2$-approval}
    \label{fig:nw_voters_2app}
    \end{subfigure}
    \begin{subfigure}{0.3\textwidth}
    \centering
    \includegraphics[width=\linewidth]{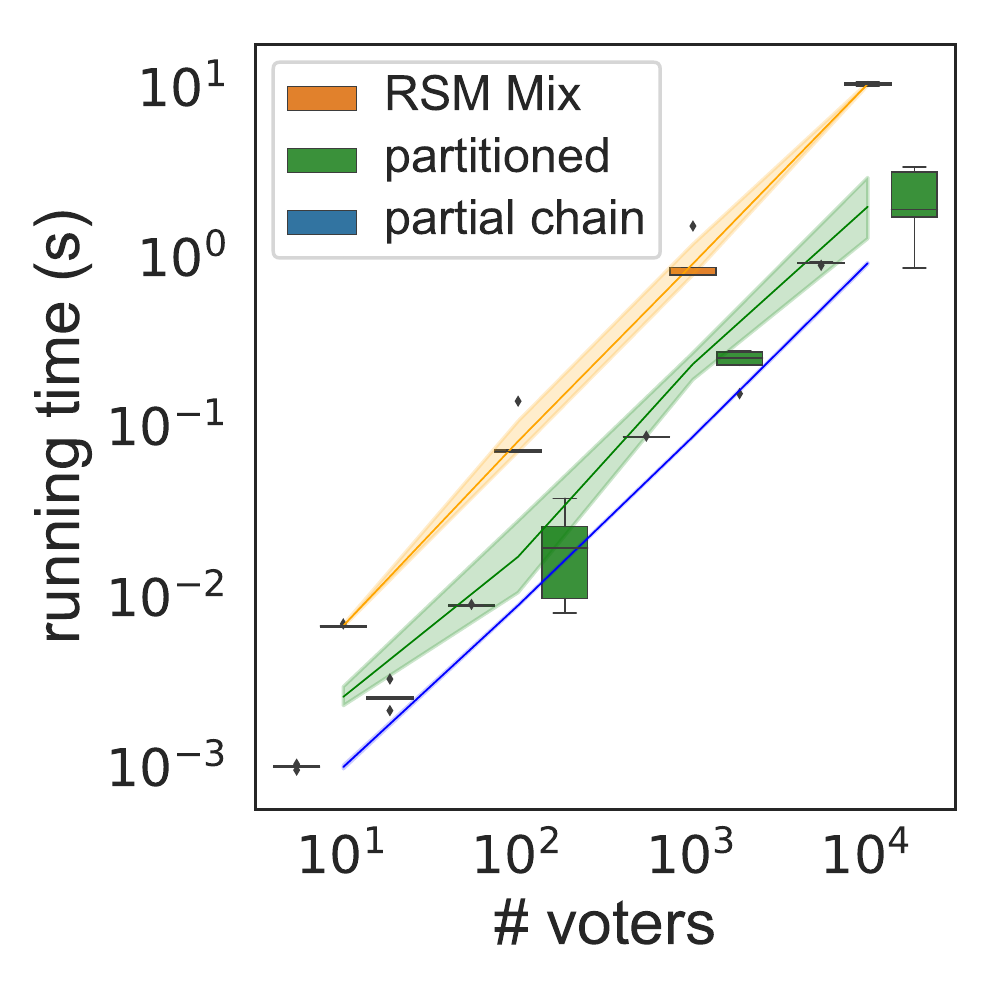}
    \caption{Borda}
    \label{fig:nw_voters_Borda}
    \end{subfigure}
    \caption{Running time of the {\bf necessary winners} computation for 100 candidates and between 10 and 10,000 voters, for three positional scoring rules.}
    \label{fig:nw_voters}
\end{figure}

Next, we analyzed the speed-up achieved by the optimized necessary winners algorithm for $n=10,000$ voters, with $m$ ranging from $10$ to $200$ on a linear scale. Figure~\ref{fig:nw_candidates} shows these results in comparison to a baseline, where we reuse computation of $\textsc{Up}_{P}(c)$ and $\textsc{Down}_{P}(c)$  across candidates, but do not re-order candidates in a competition, and also do not optimize the computation of $\textsc{Up}_{P}(c)$ and  $\textsc{Down}_{P}(c)$  based on the structure of $\P$.  We observe that the optimized implementation outperforms the baseline by a factor of 10-20 in most cases.  Overall, speed-up improves with increasing number of candidates, and  partial chains and partitioned preferences datasets show the highest speed-up.  We see significant variability in Figure~\ref{fig:nw_candidates_plurality}, because some of the instances had  necessary winners and others did not.

We also analyzed the running time and \colorTwo  observed that this computation is efficient: RSM Mix completes in under 40 seconds for $m=200$ for \plurality and  $2$-approval, and for all except one case of Borda, where it takes 60 seconds.  For  partitioned preferences and partial chains, the computation completes in under 8.5 seconds and 2 seconds, respectively, pointing to the effectiveness of the optimizations that use the structure of $P$.  

Finally, the running times were interactive on real datasets: 0.006 seconds for all scoring rules on {\em dessert}, and 0.28 seconds on {\em travel}. We achieved a factor of 2-2.5 speed-up over the baseline version for {\em dessert}, and a factor of 5-6.7 speed-up for {\em travel}.  Speed-up was most significant for Borda, with running time decreasing from 1.87 seconds to 0.28 seconds.

\begin{figure}[t!]
    \begin{subfigure}{0.3\textwidth}
    \centering
    \includegraphics[width=\linewidth]{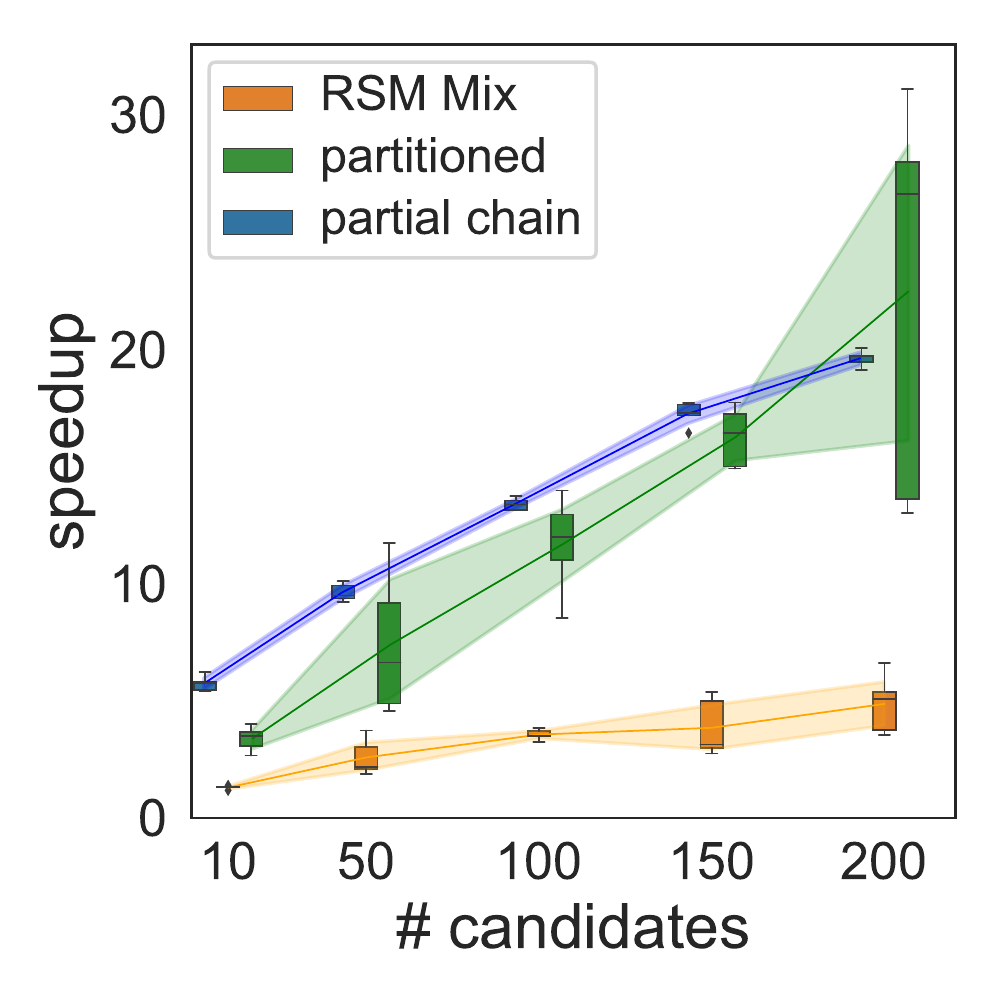}
    \caption{plurality}
    \label{fig:nw_candidates_plurality}
    \end{subfigure}\quad
    \begin{subfigure}{0.3\textwidth}
    \centering
    \includegraphics[width=\linewidth]{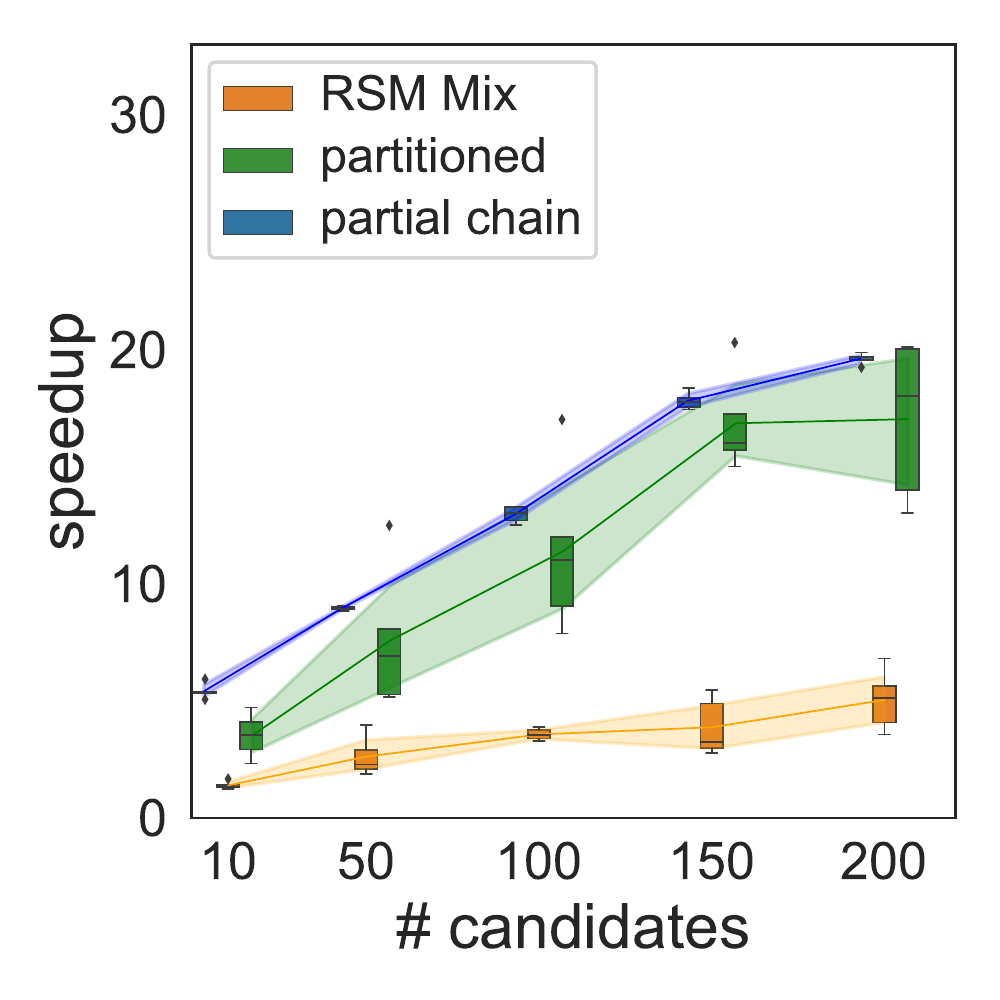}
    \caption{$2$-approval}
    \label{fig:nw_candidates_2app}
    \end{subfigure}\quad
    \begin{subfigure}{0.3\textwidth}
    \centering
    \includegraphics[width=\linewidth]{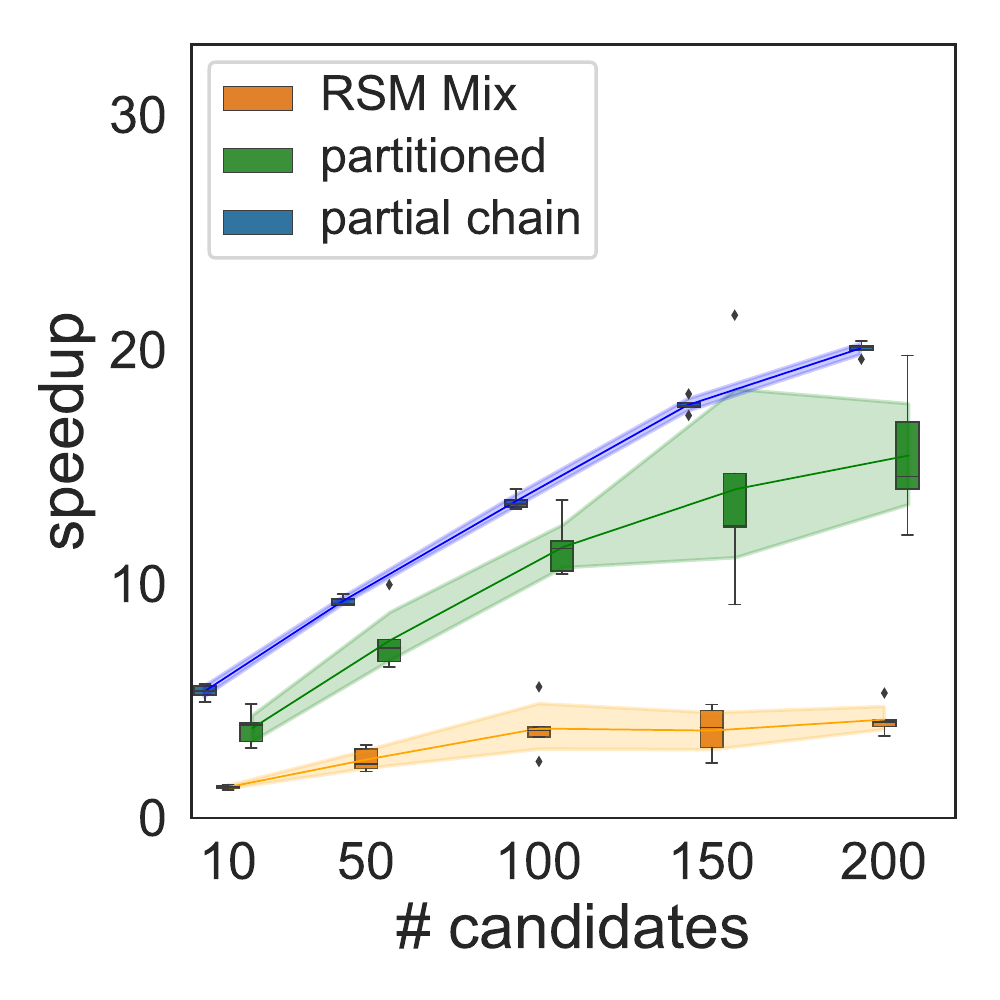}
    \caption{Borda}
    \label{fig:nw_candidates_Borda}
    \end{subfigure}
    \caption{Speed-up factor of the optimized {\bf necessary winners} computation over the baseline, for 10,000 voters, and between 10 and 200 candidates, for three positional scoring rules.  Performance is improved by a factor of 10-20 in most cases, and is highest for partial chains and partitioned preferences.}
    \label{fig:nw_candidates}
\end{figure}

\subsection{Possible Winners}
\label{sec:exp:pw}

In this section, we evaluate the performance of appropriate methods for the computation of possible winners (PW) under \plurality, $2$-approval, and Borda.


\paragraph{Plurality}  To compute PW under \plurality, we implemented an optimized version of the polynomial-time algorithm by Betzler and Dorn~\cite{DBLP:journals/jcss/BetzlerD10}, as described in Section \ref{sec:pw-betzler-dorn}.   Figure~\ref{fig:PW_Plurality} shows the running time of our implementation.  In Figure~\ref{fig:PW_Plurality}(a), we set $m=25$ and vary $n$ between 10 and 10,000 on a logarithmic scale, while in Figure~\ref{fig:PW_Plurality}(b) we set $n=10,000$ and vary $m$ between 5 and 25 on a linear scale.  We observe that this algorithm is efficient: most instances complete in less than 0.5 second, with the exception of a single instance that takes just over 1 second. The running time is higher when there are more possible winners.  For this reason, computation is fastest on partitioned preferences datasets, and slowest on partial chains datasets.  (Note that we set $m$ lower for PW experiments than for NW, where $m$ went up to 200, to have the same experimental setting for \plurality as for 2-approval and Borda, presented later in this section.  A high value of $m$ is infeasible for the latter rules because of the intrinsic complexity of the problem.)

\paragraph{$2$-approval and Borda using three-phase computation}  
As discussed in Section \ref{sec:preliminaries},
computing PW  under the  Borda rule is
NP-complete both for voting profiles consisting of partial chains and for voting profiles consisting of partitioned preferences.  Furthermore, computing PW under $2$-approval is NP-complete for voting profiles consisting of partial chains, but is polynomial-time solvable for voting profiles consisting of partitioned preferences.
In view of the intractability implied by the aforementioned NP-complete cases,  we use the three-phase method described in Section~\ref{sec:possible:threephase} that may invoke the ILP solver for difficult cases. We evaluate  the performance of this method here, demonstrating the impact of the  number of voters $n$ and the number of candidates $m$ on the running time of PW. (Note that we include $2$-approval for partitioned preferences into the comparison, for consistency of presentation.)  We fix $m$ at 25 and vary $n$ between 10 and 10,000 on a logarithmic scale, and then fix $n$ at 10,000 and vary $m$ between 5 and 25 on the linear scale. Even with these modest values of $m$, the ILP solver can take a very long time.  Thus, to make our experimental evaluation manageable, we set an end-to-end cut-off of 2,000 seconds per instance.  In what follows, we report the running times of the instances that completed within the cut-off, and additionally report the percentage of completed instances.

Figures~\ref{fig:PW_2app}(a) and \ref{fig:PW_Borda}(a) show the running time as a function of the number of voters for $2$-approval and Borda, respectively. The running time increases linearly with the number of voters.  Interestingly,  partial chains datasets take as long or longer to process as  RSM Mix datasets.  This is because Phase 1 of the three-phase computation is more effective for 
RSM Mix, with fewer candidates passed on to Phase 2.
Phases 1 and 2 are effective in pruning non-winners and in identifying clear possible winners. Of the 75 instances we executed for this experiment for each scoring rule, only 13 (17\%) needed to execute Phase 3 (\ie invoke the ILP solver) for 2-approval, and only 9 (12\%) ---  for Borda, with at most 3 candidates to check.  Of the 9 instances that reached Phase 3 for Borda, 6 timed out at 2,000 sec.  No other instances timed out in this experiment.  Instances reaching Phase 3 are responsible for the high variability in the running times. For 2-approval, all instances that reached Phase 3 computed in under 148 seconds (median 8.46 seconds, mean 30.69 seconds, stdev 48.72 seconds).  For Borda, for the three instances that executed Phase 3 and did not timeout, the running times were 6 sec, 15 sec, and 381 sec.  In contrast, all remaining instances --- those that did not execute Phase 3 --- computed in under 53 sec (median 2.44 sec, mean 9.50 sec, stdev 15.14 sec).

\begin{figure}[t]
  \centering
  \begin{tabular}{c @{} c }a
  \hspace{3mm}\includegraphics[scale=.475]{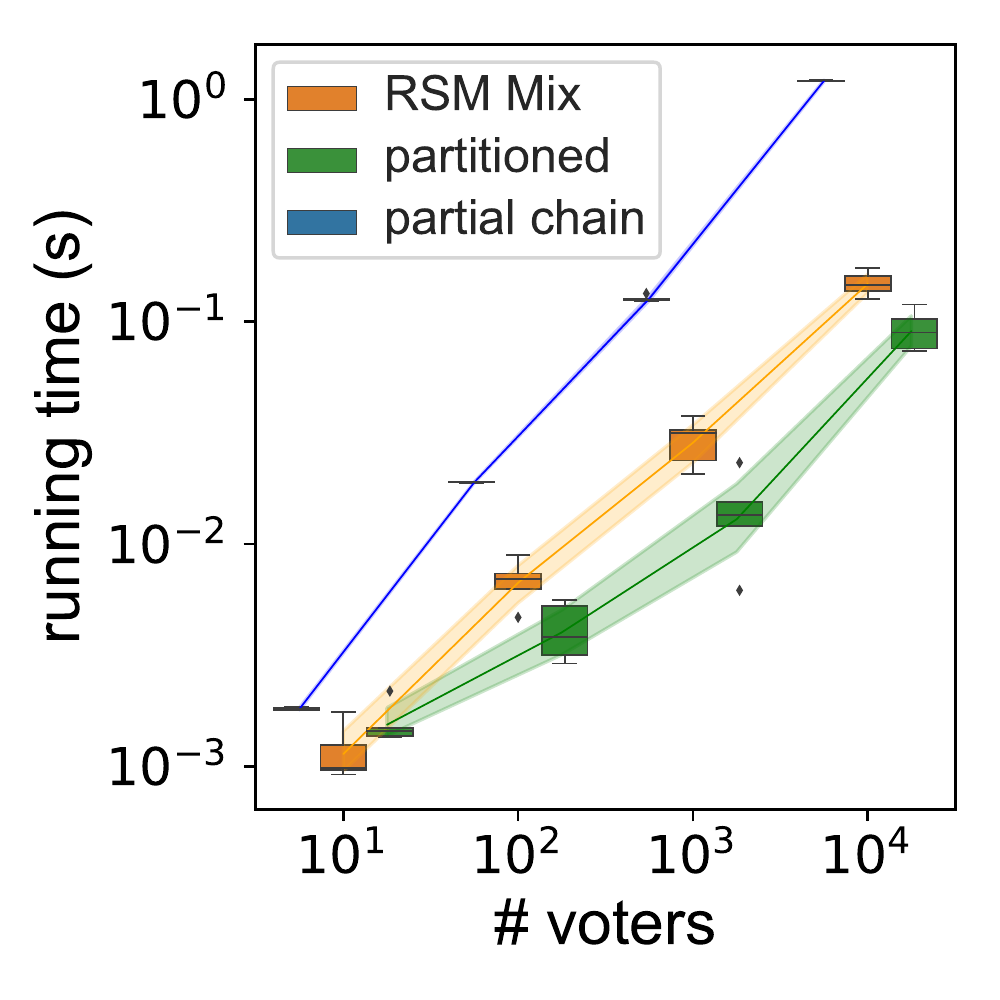} \hspace{3mm} &
  \includegraphics[trim=0 0 0 0, clip, scale=.4750]{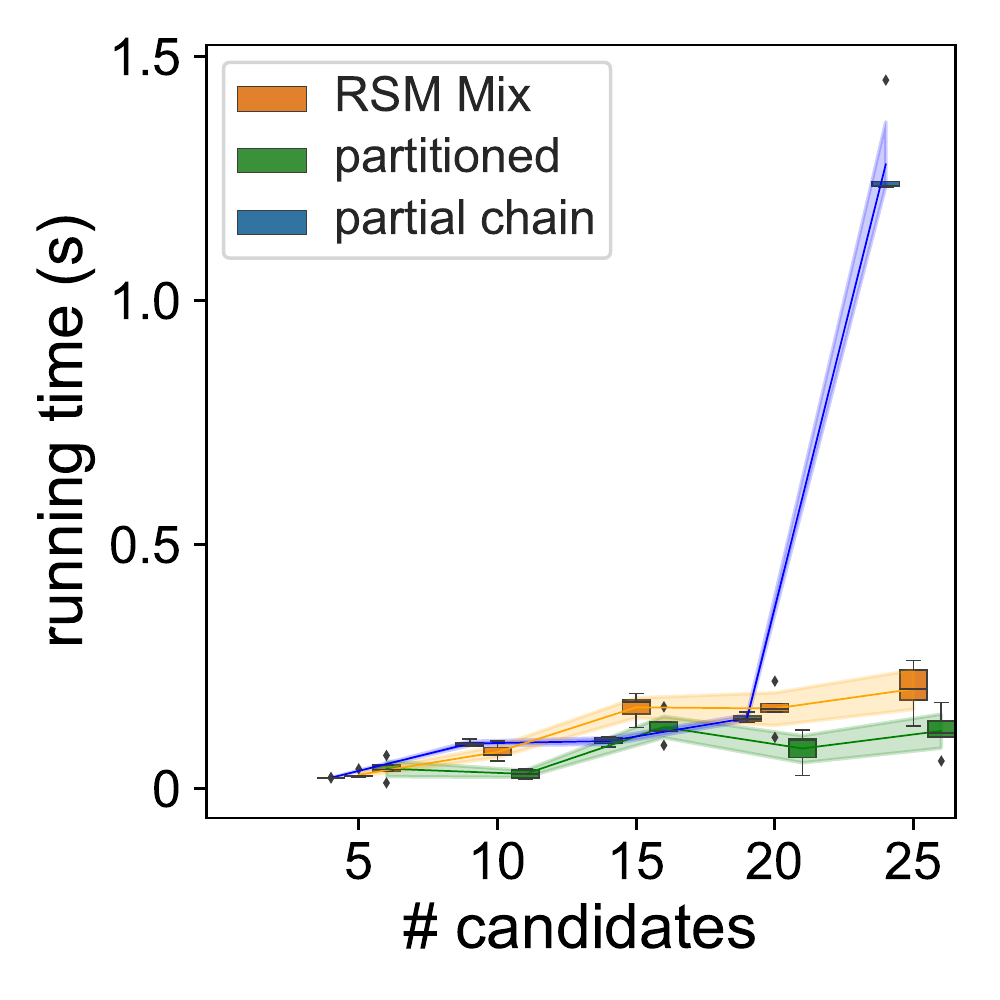}\\
  \small (a) 25 candidates & \small (b) 10,000 voters
  \end{tabular} 
  \caption{Running time of the computation of the set of {\bf possible winners} on \textbf{\plurality},  using an optimized implementation of the algorithm from Betzler and Dorn~\cite{DBLP:journals/jcss/BetzlerD10}. Most instances complete in less than 0.5 seconds. Running times are higher when there are more possible winners.  For this reason, computation is fastest on partitioned preferences, and slowest on partial chains. 
  }
  \label{fig:PW_Plurality}
\end{figure}


Figures~\ref{fig:PW_2app}(b) and~\ref{fig:PW_Borda}(b) show the running times as a function of the number of candidates under 2-approval and Borda.  We make similar observations here as in our discussion of Figures~\ref{fig:PW_2app}(a) and~\ref{fig:PW_Borda}(a), noting that only 9 instances instances reached Phase 3 for 2-approval, and only 4 reached Phase 3 for Borda.  These instances, all RSM Mix, took longer to run, and contributed the most to running time variability.

\paragraph{$2$-approval on partitioned preferences: three-phase computation vs. network flow} 
\eat{
\julia{Phokion, I tried moving this, so this paragraph is before ``$2$-approval and Borda under three-phase computation'', per our discussion, but I didn't like it there, because what's here is really a detail, we should present the more important results first.  Also, in this paragraph I need to know what three-phase computation is, and the reader is reminded of this in the earlier set of experiments.  I don't feel strongly about this, please move if you like.}}
For the $2$-approval rule on voting profiles consisting of  partitioned preferences, we also implemented Kenig's~\cite{DBLP:conf/atal/Kenig19}  polynomial-time algorithm, which is based on network-flow, and we compared its performance to that of three-phase computation.   
Figure~\ref{fig:PW_2app}(c) shows the running times as a function of the number of partitions for instances containing 25 candidates and 10,000 voters. None of the 30 instances needed ILP (Phase 3) while using the three-phase computation. Overall, our three-phase approach is both more general in terms of the datasets it handles, and it outperforms the polynomial-time network-flow algorithm for partitioned preferences.  

\begin{figure}[t]
  \centering
  \begin{tabular}{c @{} c @{} c }
  \hspace{-3mm}\includegraphics[scale=.475]{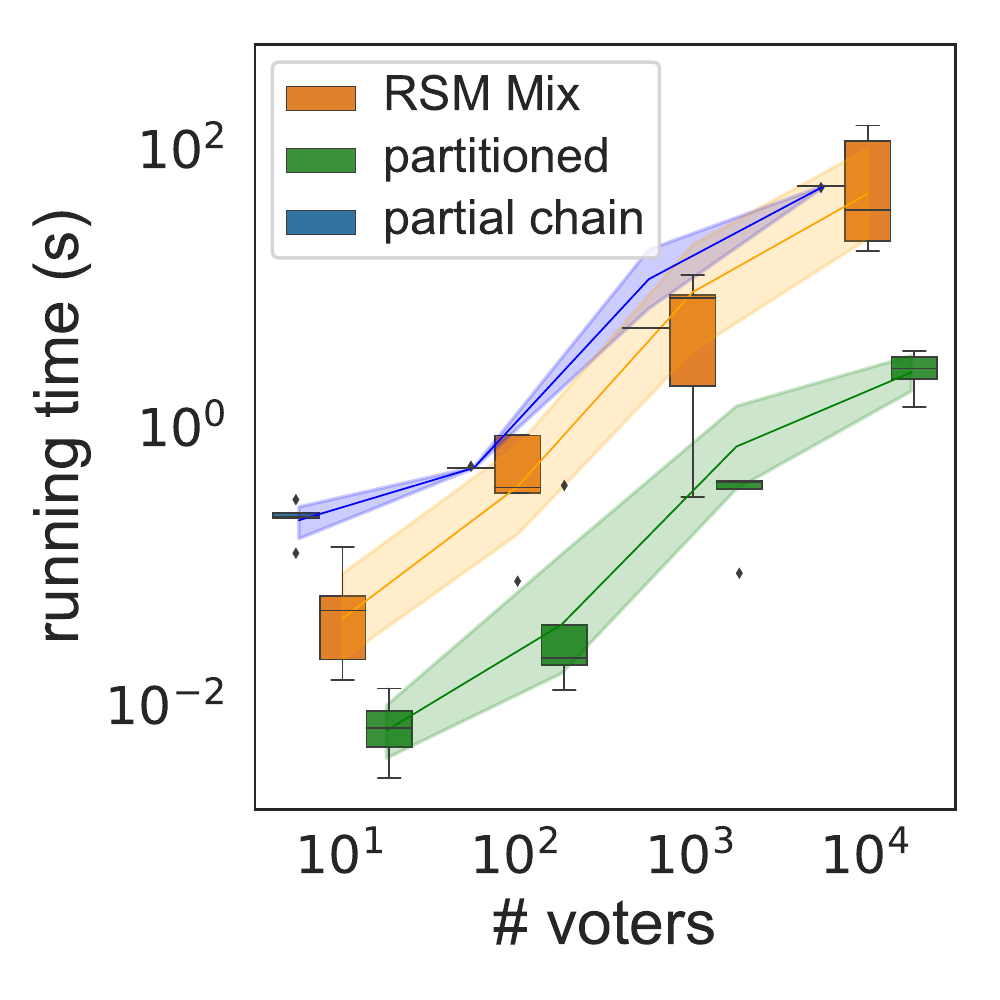} \hspace{-2mm} &
  \includegraphics[trim=0 0 0 0, clip, scale=.4750]{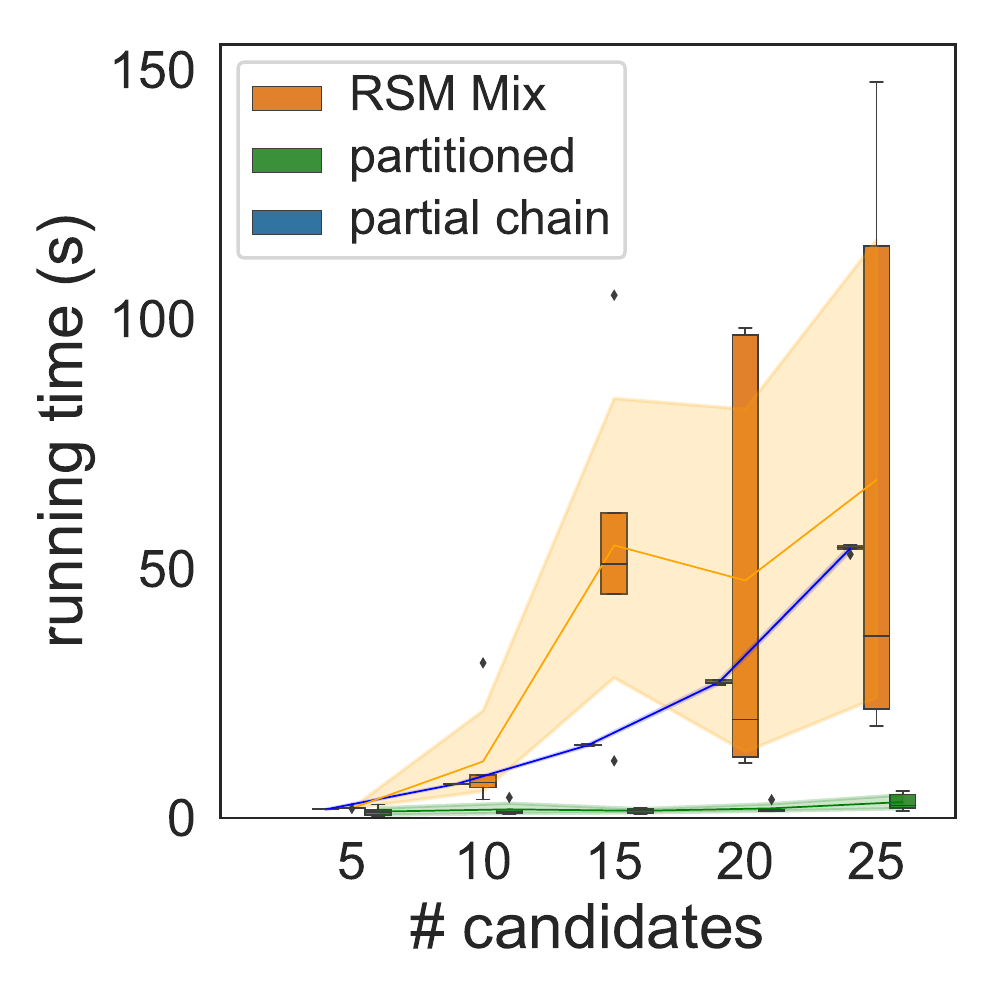}&
  \includegraphics[trim=0 0 0 0, clip, scale=.4750]{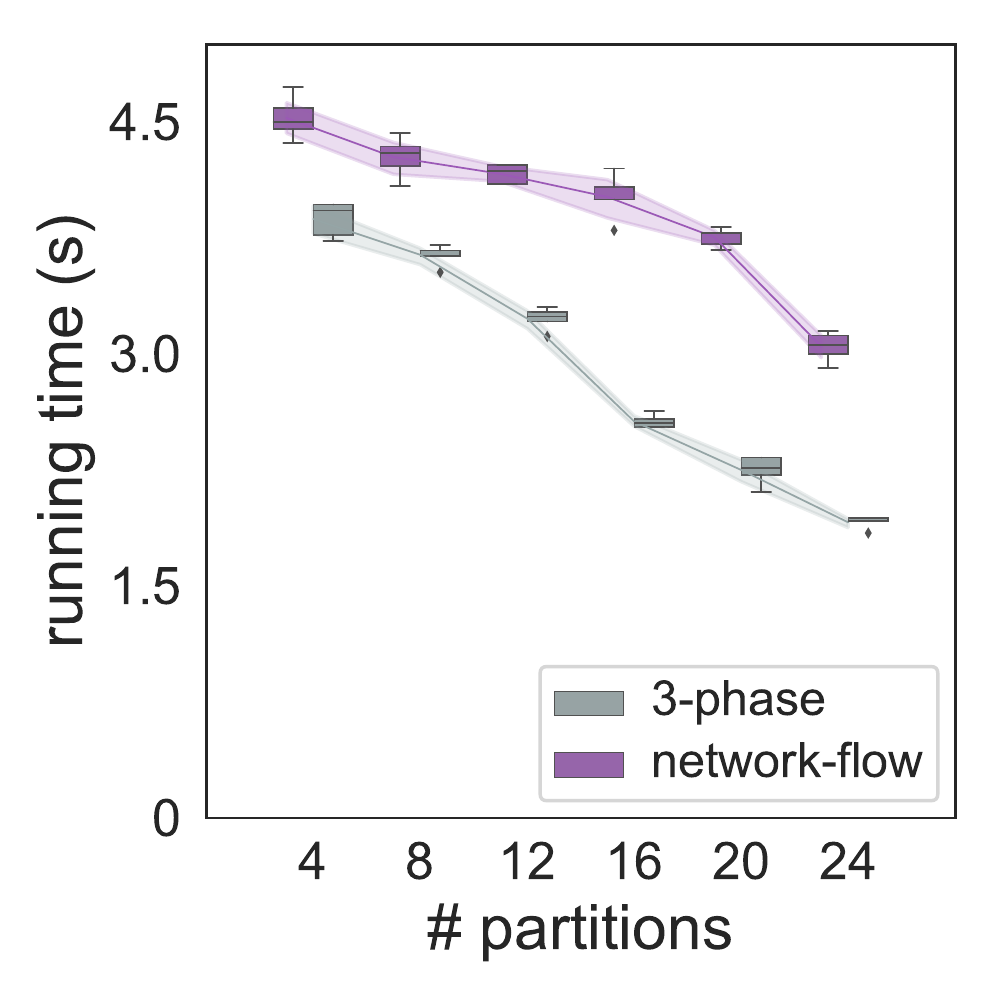}\\
  \small (a) 25 candidates & \small (b) 10,000 voters & \small (c) 25 candidates \& 10,000 voters
  \end{tabular} 
  \caption{Running time of the three-phase computation of the set of {\bf possible winners} on \textbf{2-approval} scoring rule. (a) None of the 75 instances timed out at 2,000 sec.  These were among only 13 instances (9 RSM Mix and 4 partial chain) that needed to execute the ILP solver in Phase 3. (b) None of the 75 instances timed out at 2,000 sec. These were among only 9 instances (all RSM Mix) that needed to execute the ILP solver in Phase 3. (c) Comparison of three-phase computation with polynomial-time algorithm (based on flow network and theoretical results in \cite{DBLP:conf/atal/Kenig19}) by varying the number of partitions in each instance.
  }
  \label{fig:PW_2app}
\end{figure}

\begin{figure}[t]
  \centering
  \begin{tabular}{c @{} c}
  \hspace{3mm}\includegraphics[scale=.475]{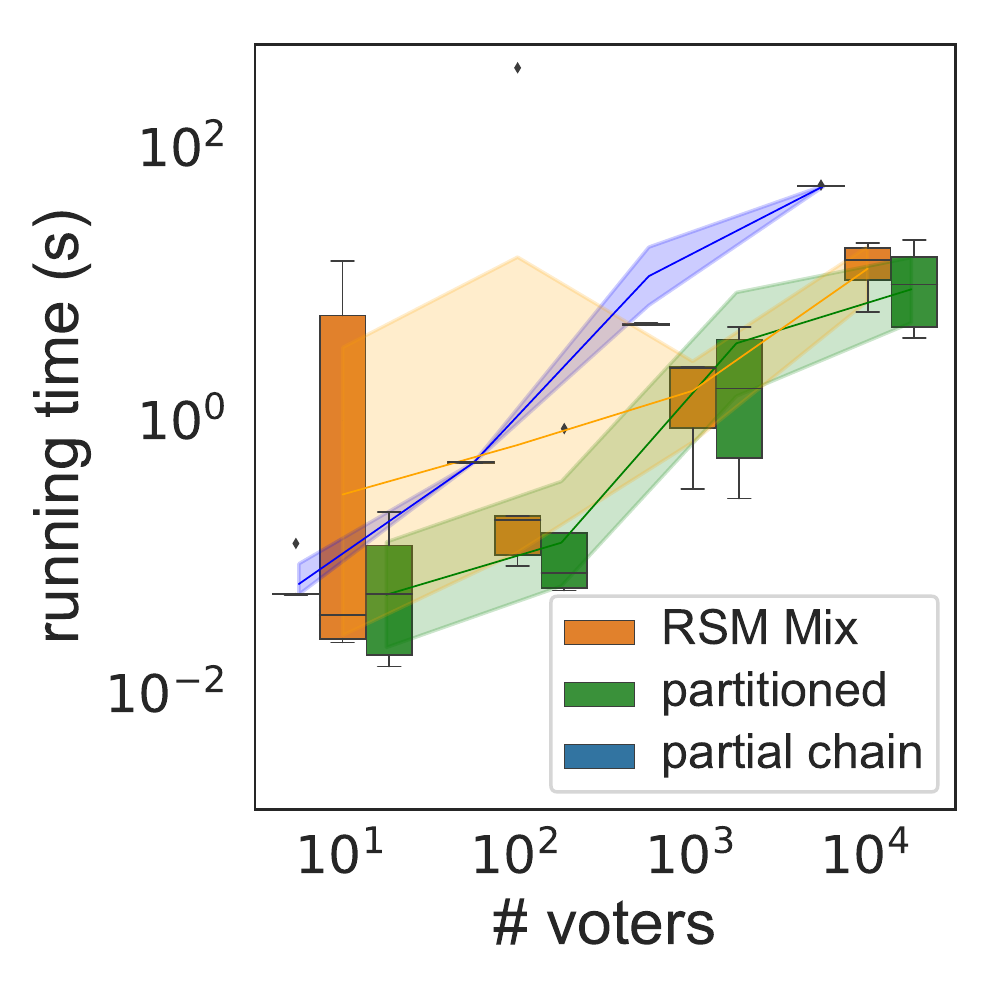} \hspace{3mm} &
  \includegraphics[trim=0 0 0 0, clip, scale=.4750]{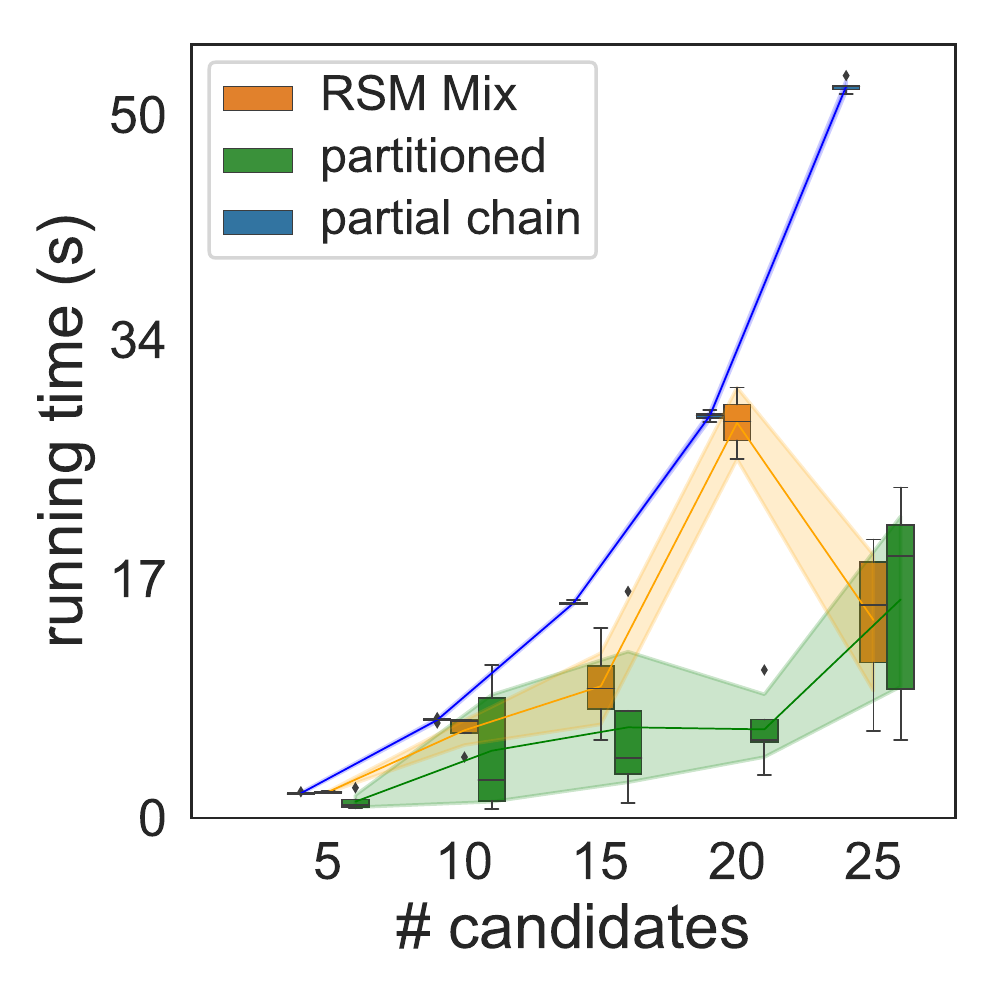}\\
  \small (a) 25 candidates & \small (b) 10,000 voters
  \end{tabular} 
  \caption{Running time of the three-phase computation of the set of {\bf possible winners} on \textbf{Borda} scoring rule. (a) 6 (all RSM Mix) out of 75 instances (8.0\%) timed out at 2,000 sec.  These were among only 9 instances (all RSM Mix) that needed to execute the ILP solver in Phase 3. (b) 4 (all RSM Mix) out of 75 instances (5.3\%) that needed to execute the ILP solver, timed out at 2,000 sec.
  }
  \label{fig:PW_Borda}
\end{figure}

\paragraph{Drilling down on the phases of the three-phase computation}
Next, we measured the effectiveness of the first two phases of the three-phase computation, which run in polynomial time in the number of candidates.  To do so, we calculate the proportion of profiles for which the three-phase computation terminates after the first two phases, under the Borda scoring rule. We created 10,000 profiles consisting of 100 voters and 10 candidates using a mixture of three RSMs, as described in Section~\ref{sec:exp:setup}. 

Figure~\ref{fig:RSM_sucess_dataset} presents the density distribution of the resulting posets (as in Figure~\ref{fig:RSM_eval_gehrlein}), and highlights the instances for which PW computation terminated after two phases in purple, and those for which all three phases were necessary in yellow.   In summary, PW terminated after two phases for 91.62\% of the instances.  Phase 3 was needed primarily  when the $\phi$ parameter was low and the average density was medium, or when the $\phi$ parameter as well as the average density were both high. Profiles with low average density always terminated after the second phase. 

We also compared the average running time of the first two phases of the PW algorithm using RSM profiles with profiles generated using Gehrlein's methods.   In this experiment, we generated 10,000 profiles using a mixture of 3 RSM models, with $m=10$ candidates and $n=100$ voters, and used the Borda scoring rule.  RSM profiles generally take more time across different values of $\phi$ (Figure~\ref{fig:running_time_success}(a)) and across different poset densities (Figure~\ref{fig:running_time_success}(b)) as RSM is more generalized  (Figure~\ref{fig:RSM_eval_gehrlein}).  This finding once again underscores that RSM is able to generate interesting posets, which may be more challenging to process than those generated with alternative methods.

\paragraph{PW on real datasets}
Finally, we computed PW for the real datasets {\em dessert} and {\em travel} using three-phase computation, and found multiple possible winners for all scoring rules.  In all cases, winners were determined in Phases 1 and 2 of the computation, and the ILP is never invoked. All executions took under 23 seconds.  
{\colorOne}

\colorOne


 
 

\colorTwo

\begin{figure}[t]
    \centering
    \includegraphics[scale=0.475]{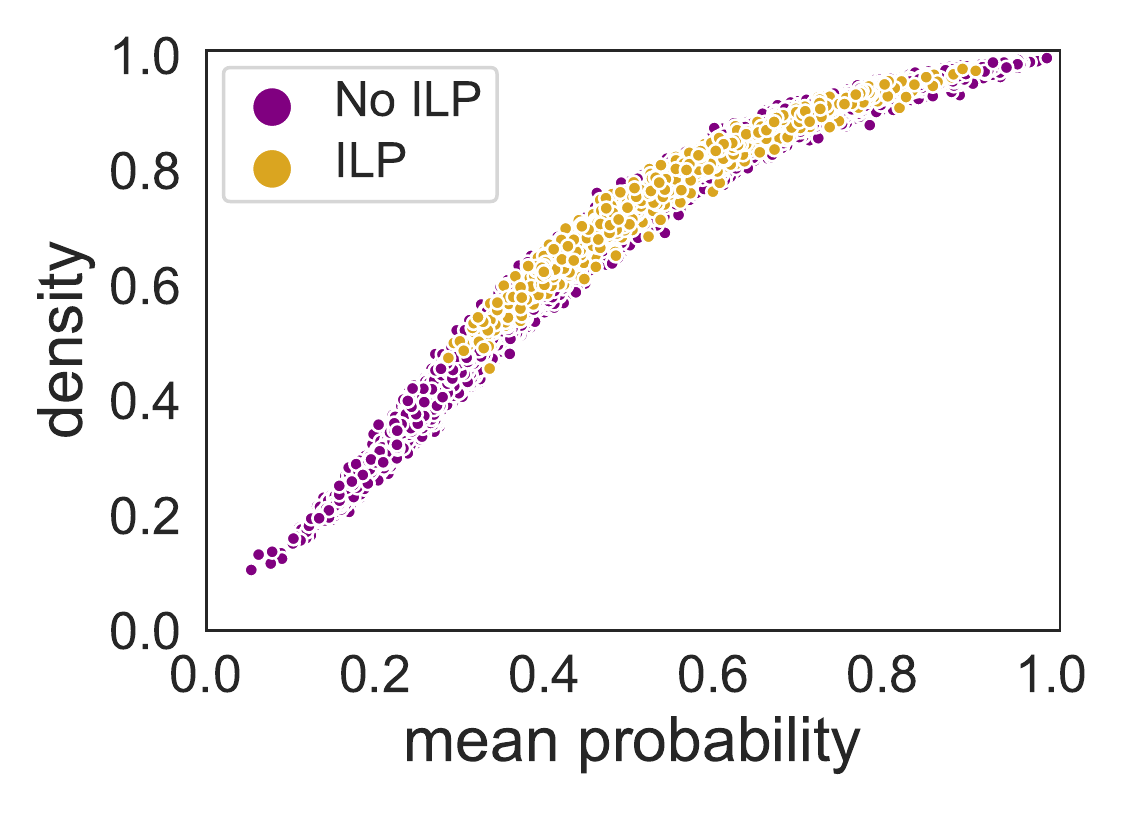}
  \caption{91.62\% of 10,000 RSM profiles, with 10 candidates and 100 voters, found the entire set of possible winners under the Borda scoring rule after the first two phases of pruning without the need of using ILP.}
  \label{fig:RSM_sucess_dataset}
\end{figure}

\begin{figure}[t]
  \centering
  \begin{tabular}{c @{} c }
  \hspace{3mm}\includegraphics[scale=.475]{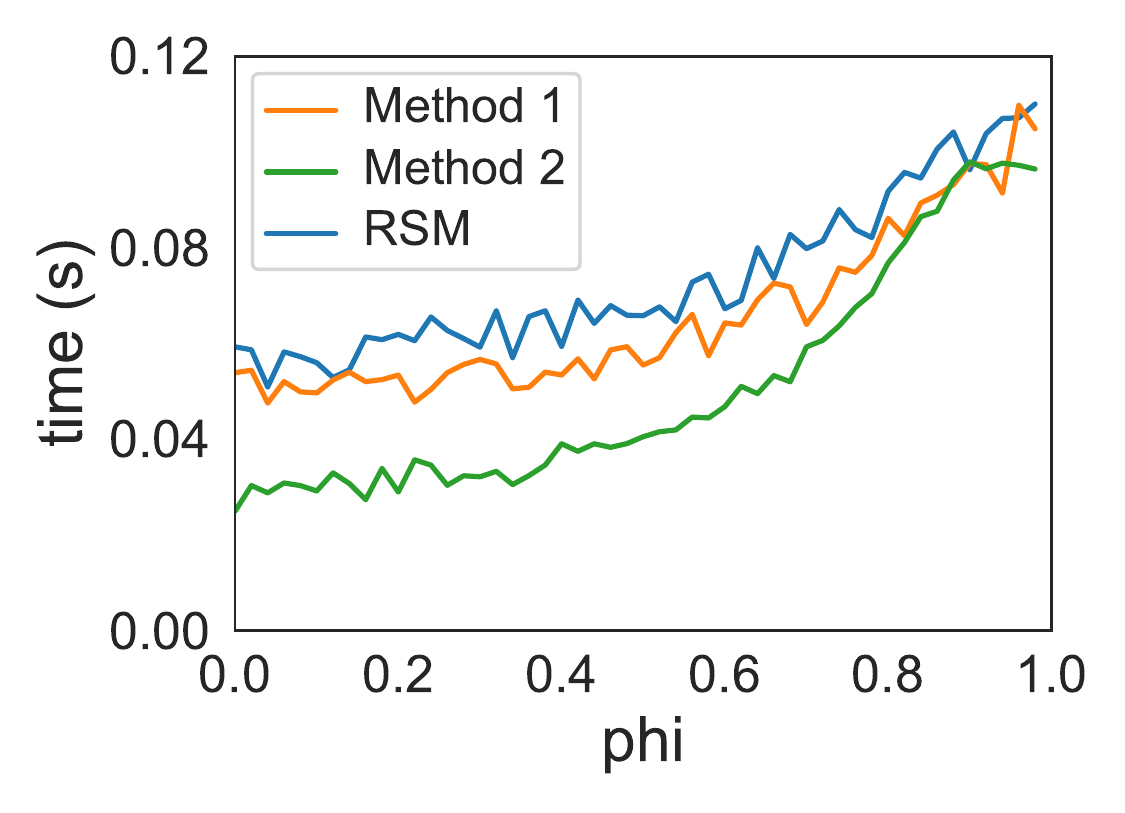} \hspace{3mm} &
  \includegraphics[trim=0 0 0 0, clip, scale=.4750]{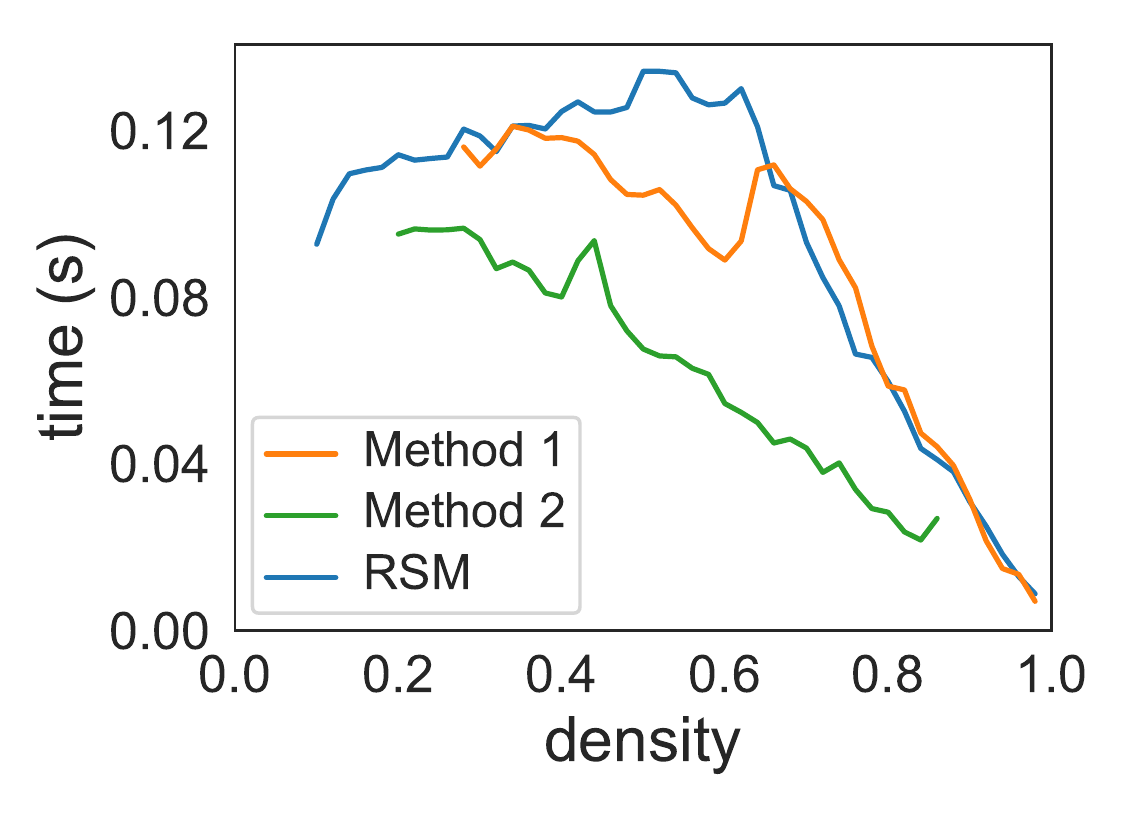}\\
   \small (a) & \small (b) 
  \end{tabular} 
  \caption{Average running times of the first two phases of three-phase computation, for  10,000 RSM profiles, with 10 candidates and 100 voters, for the Borda scoring rule.  Method 1 and Method 2 are from Gehrlein~\cite{gehrlein1986methods}. RSM generates challenging profiles for the computation of PW.}
  \label{fig:running_time_success}
\end{figure}

\section{Concluding Remarks}

The contributions made in this paper can be summarized as follows.
\begin{itemize}
    \item  We introduced new methods for generating partial orders that are
of interest in their own right, most notably, the Repeated Selection Model.
    \item Furthermore, we produced a rich set of datasets that can serve as benchmarks  in other experiments concerning incomplete preferences in computational social choice.
\item We presented a number of algorithmic techniques for computing the necessary winners and the possible winners for positional scoring rules in the presence of incomplete preferences.
 We demonstrated that our techniques scale well in a variety of settings, including settings in which computing the possible winners is an NP-hard problem.

\end{itemize}

 The algorithmic techniques and the data generation methods presented here may find applications in
other frameworks, including the framework introduced in \cite{kimelfeld2018computational} and studied further in \cite{DBLP:conf/pods/KimelfeldKT19}, which aims to bring together computational social choice and databases by supporting queries about winners in elections together with relational context about candidates, voters, and candidates' positions  on issues.

\bibliographystyle{unsrt}  
\bibliography{comsoc}  


\end{document}